\documentclass[journal]{IEEEtran}

\usepackage{float} 
\usepackage{amsmath,graphicx,amssymb,color} 
\usepackage{algorithmicx}
\usepackage{algpseudocode}
\usepackage{breqn}
\usepackage{cite}
\usepackage{amsthm}

\newtheorem{propo}{Proposition}

\DeclareMathOperator*{\argmax}{arg\,max}

\newcommand*\diff{\mathop{}\!\mathrm{d}}

\usepackage{multirow}
\usepackage[caption=false,font=footnotesize,labelfont=footnotesize]{subfig}
\usepackage{bbm} 
\usepackage{mathtools} 
\usepackage{enumerate} 
\usepackage{flushend} 
\usepackage{stmaryrd} 

\DeclarePairedDelimiter{\floor}{\lfloor}{\rfloor}
\usepackage{cite}
\usepackage{accents}

\DeclareMathOperator{\LL}{L_2}

\usepackage{mathrsfs}
\usepackage{todonotes}

\usepackage{multicol} 

\begin{document}

\title{Domain-Informed Spline Interpolation}

\author{Hamid~Behjat,~\IEEEmembership{Member,~IEEE}, Zafer~Do\u{g}an,~\IEEEmembership{Member,~IEEE}, Dimitri~Van~De~Ville,~\IEEEmembership{Senior~Member,~IEEE}, and Leif~S\"ornmo,~\IEEEmembership{Fellow,~IEEE}
\thanks{
This work was supported by the Swedish Research Council under Grant 2009-4584, and in part by the Swiss National Science Foundation under grant P2ELP2-165160.
}
\thanks{
H. Behjat and L. S\"ornmo are with the Department of Biomedical Engineering, Lund University, Sweden (e-mail: hamid.behjat@bme.lth.se; leif.sornmo@bme.lth.se).
} 
\thanks{
Z. Do\u{g}an is with the School of Engineering and Applied Sciences, Harvard University, USA (e-mail: zaferdogan@seas.harvard.edu ).
} 
\thanks{
D. Van De Ville is with the \'Ecole Polytechnique F\'ed\'erale de Lausanne, Switzerland (e-mail: dimitri.vandeville@epfl.ch) and the Department of Radiology and Medical Informatics, University of Geneva.
}
\thanks{1053-587X (c) 2019 IEEE --- This article has been accepted for publication in IEEE Transactions on Signal Processing. Personal use is permitted, but republication/redistribution requires IEEE permission. Link to article: https://ieeexplore.ieee.org/document/8734789 --- Citation info: DOI 10.1109/TSP.2019.2922154, IEEE Transactions on Signal Processing.
}
}


\maketitle

\begin{abstract}
Standard interpolation techniques are implicitly based on the assumption that the signal lies on a single homogeneous domain. In contrast, many naturally occurring signals lie on an inhomogeneous domain, such as brain activity associated to different brain tissue. We propose an interpolation method that instead exploits prior information about domain inhomogeneity, characterized by different, potentially overlapping, subdomains. As proof of concept, the focus is put on extending conventional shift-invariant B-spline interpolation. Given a known inhomogeneous domain, B-spline interpolation of a given order is extended to a domain-informed, shift-variant interpolation. This is done by constructing a domain-informed generating basis that satisfies stability properties. We illustrate example constructions of domain-informed generating basis, and show their property in increasing the coherence between the generating basis and the given inhomogeneous domain. By advantageously exploiting domain knowledge, we demonstrate the benefit of domain-informed interpolation over standard B-spline interpolation through Monte Carlo simulations across a range of B-spline orders. We also demonstrate the feasibility of domain-informed interpolation in a neuroimaging application where the domain information is available by a complementary image contrast. The results show the benefit of incorporating domain knowledge so that an interpolant consistent to the anatomy of the brain is obtained. 

\end{abstract}

\begin{IEEEkeywords}
Sampling, interpolation, context-based interpolation, B-splines, multi-modal image interpolation.
\end{IEEEkeywords}

\IEEEpeerreviewmaketitle

\section{Introduction}
\IEEEPARstart{I}{nterpolation} has been extensively studied in various settings. The main frameworks are based on concepts such as smoothness for spline-generating spaces~\cite{Unser1999}, underlying Gaussian distributions for ``kriging''~\cite{Kriging1990}, and spatial relationship for inverse-distance-weighted interpolation~\cite{Lu2008}. Yet, while advanced concepts have been developed for describing these signal spaces, the underlying domain is always assumed to be homogeneous. Here, we consider a different scenario in which signals are sampled over a known ``inhomogeneous'' domain; i.e., a domain characterized by a set of subdomains with their description available as supplementary data. In a sub-catagory of super-resolution image processing techniques, the interpolation phase is adapted such that the scheme becomes close to the interpolation scenario that we consider here. In particular, information from a high resolution signal is exploited to enhance interpolation on a set of samples acquired at low resolution. For example, an iterative, patch-based non-local reconstruction scheme was introduced in \cite{Manjon2010b}; the use of inter-modality priors to regularize the similarity between the up-sampled image and a secondary high resolution image was proposed in \cite{Rousseau2010}; a non-local means feature-based technique that uses structural information of a high resolution image with a different contrast was presented in \cite{JafariKhouzani2014}; sparsity promoting priors using overcomplete dictionaries learned from the data were exploited in \cite{Rueda2013}. Yet, such schemes do not fall within the problem that we formulate in this paper. The low resolution data samples and the supplementary high resolution information considered in these proposals are both of the same nature and, in essence, both describe the \textit{signal}. In the scheme proposed in this article, the supplementary high resolution information instead describes the {\it domain} of the signal, which has a completely different temporal/spatial characteristic than that of the signal (samples) defined on the domain. Moreover, rather than being a learning scheme, our proposal is formulated as a shift-variant extension of conventional shift-invariant interpolation schemes \cite{Unser1993,Unser2000}, such as linear, cubic or higher order B-spline interpolations. 

\subsection*{Problem Formulation}
\label{sec:probform}
Assume that the following set of information is given:
\begin{enumerate}
\item Domain knowledge described by a set of non-negative subdomain functions
\begin{align}
\label{eq:subdomains}
\mathcal{D}:\bigg \{d_j \in \text{L}_{2}, \: j \in &\mathcal{J}:\{1,\ldots,J\} \bigg \}, 
\end{align}
which form a partition of unity as
\begin{equation}
\label{eq:pouDomain}
\sum_{j=1}^J d_j(x) = 1, \quad \forall x \in \mathbb{R}.
\end{equation}

\subsubsection*{Remark} We denote a segment $[a,b], a,b \in \mathbb{R}$, of the domain as \textit{homogeneous} if there exists an $l\in\mathcal{J}$ such that 
\begin{align}
\label{eq:homoD_1}
d_{l}(x) & = 1, \quad \forall x \in [a,b],
\end{align}
and \textit{inhomogeneous} if (\ref{eq:homoD_1}) does not hold for any $l\in \mathcal{J}$. 

\item
A uniform sequence of samples $s[k]$ taken from a continuous function $s(x)$ as 
\begin{align}
\label{eq:samples}
s[k]  & = \left \langle s(x), \delta(x-kT) \right  \rangle,
\end{align}
where $T \in \mathbb{R}^{+}$ denotes the sampling step. 
\end{enumerate}

\noindent
The objective is to recover $s(x)$ using the given set of samples $s[k]$. The approach is to extend any conventional interpolation method that employs a shift-invariant basis to an interpolation that employs a shift-variant basis, such that prior information on the domain inhomogenuity is accommodated. In particular, consider a compactly-supported generator $\varphi(x)$ (i.e., $\varphi(x)=0$, $\forall |x| \ge \: \Delta^{(\varphi)} \in \mathbb{R}^{+}$) of a shift-invariant space 
\begin{equation}
\label{eq:representation0}
\mathcal{V}_{\varphi} = \bigg\{ \tilde{s}_{T}(x) = \sum_{k\in \mathbb{Z}} {c}[k] \cdot \varphi\left(\frac{x}{T}-k\right) : {c} \in \ell_{2} \bigg\}.
\end{equation}
The generating function $\varphi(x)$ can be any of the compact-support kernels used in standard interpolation. In the presence of domain inhomogeneity, the idea is to transform $\varphi(\cdot-k)$ into $\varphi_{k}(\cdot-k)$: a tailored version of $\varphi(\cdot-k)$ whose definition is based on the domain structure in the adjacency of $k$. With this construction, a shift-variant, domain-informed space
\begin{equation}
\label{eq:representation}
\mathcal{V}^{(\mathcal{D})}_{\varphi} = \bigg\{ \tilde{s}_{T}(x) = \sum_{k\in \mathbb{Z}} {c}[k] \cdot \varphi_{k}\left(\frac{x}{T}-k\right) : {c} \in \ell_{2} \bigg\},
\end{equation}
is obtained. \\

\noindent
In this article, we focus on formulating the domain-informed interpolation theory for B-spline generating functions \cite{Unser1999}.\\ 

Different application areas can be envisioned for domain-informed interpolation. Earth sciences is one application area, where a spatially continuous representation of earth surface parameters, such as precipitation, land vegetation and atmospheric methane is desired to be computed from a discrete set of rain gauge measurements \cite{Borga1997,Lu2008}, fossil pollen measurements \cite{Gaillard2010,Pirzamanbein2014} and satellite estimates of methane \cite{Frankenberg2005,Keppler2006}, respectively. In these scenarios, the well-defined geographical structure of the earth, anthropogenic land-cover models \cite{Kaplan2009} and geophysical models of the earth's surrounding atmosphere may be exploited as descriptors of the inhomogeneous domain to improve the standard approach to interpolation. 

Neuroimaging is another area where domain-informed interpolation can be beneficial. In brain studies using functional magnetic resonance imaging (fMRI), a sequence of whole-brain brain functional data is acquired at relatively low spatial resolution to enable tracking brain function at high temporal resolution. An anatomical MRI scan is also commonly acquired, providing information about the convoluted brain tissue delineating gray and white matter, each of which have different functional properties \cite{Logothetis}. Unlike functional data, anatomical data can be recorded at high spatial resolution, with resolutions three- to fourfold higher than the functional data. Hence, the goal would be to exploit the richness of anatomical data to define the domain of the acquired functional data and, in turn, to improve the quality of interpolation/resampling of the data\cite{Brett2001, Andrade2001, Ashburner2007}. In this article, we will demonstrate the feasibility of domain-informed interpolation in providing a high resolution representation of brain functional data such that the representation is consistent with the underlying brain anatomy. \\

The article is organized as follows. Section~II defines the properties of domain-informed generating basis. Section~III introduces the scheme for designing a domain-informed B-spline generating basis, and presents an illustrative example construction of such a basis. Section~IV presents domain-informed B-spline interpolation. In Section~V, we show the benefit of domain-informed interpolation over standard B-spline interpolation through Monte Carlo simulations across a range of B-spline orders. In Section~VI, we conclude by demonstrating the feasibility of domain-informed interpolation in a neuroimaging application.

\section{Domain-Informed Generating Basis}
\label{sec:diib}
Let $\text{L}_{2}$ denote the Hilbert space of all continuous, real-valued functions that are square integrable in Lebesgue's sense, with the $\text{L}_{2}$ inner-product defined as 
\begin{equation}
\label{eq:L2innerproduct}
\forall f,g \in \text{L}_{2}, \quad  \left  \langle f,g  \right  \rangle_{\text{L}_{2}} = \int_{-\infty}^{+\infty} f(x) g(x) \diff x < \infty,
\end{equation}
and $\text{L}_{2}$-norm is defined for all $f \in \text{L}_{2}$ as $\| f\|_{\text{L}_{2}}^{2} =  \left  \langle f,f  \right  \rangle_{\text{L}_{2}}$. Let $\ell_{2}$ denote the Hilbert space of all discrete signals that are square summable, with the $\ell_{2}$ inner product defined as 
\begin{equation}
\label{eq:l2innerproduct}
\forall p,q \in \ell_{2}, \quad \left  \langle p,q  \right  \rangle_{\ell_{2}} = \sum_{k \in \mathbb{Z}} p[k] q[k]  < \infty,
\end{equation}
and the $\ell_{2}$-norm is defined for all $ p \in \ell_{2}$ as $\| p \|_{\ell_{2}}^{2} = \left  \langle p,p  \right  \rangle_{\ell_{2}}$. In the following, we use the convention $ \langle \cdot,\cdot \rangle = \langle \cdot,\cdot \rangle_{\ell_{2}}$ and $  \| \cdot \| = \| \cdot \|_{\ell_{2}}$ to simplify the notation.

For a given compactly-supported generator $\varphi(x)$, let $\Phi^{\ast} = \{\varphi(x/T-k)\}_{k\in\mathbb{Z}}$ denote the standard shift-invariant basis of $\varphi(x)$ and $\Phi = \{\varphi_{k}(x/T-k)\}_{k\in\mathbb{Z}}$ denote a shift-variant basis of $\varphi(x)$; i.e., for all $k \in \mathbb{Z}$,  $\varphi_{k}(\cdot-k)$ denotes a tailored version of $\varphi(\cdot-k)$, which may be identical to $\varphi(\cdot-k)$ for a subset of $k\in \mathbb{Z}$. 

\if{false}
\begin{align}
\varepsilon_{\varphi}(T) & =  \| \tilde{s}_{T}(x) - s(x)\|_{\LL}, \quad  \tilde{s}_{T}(x) \in \mathcal{V}_{\varphi} 
\\
\varepsilon^{(\mathcal{D})}_{\varphi}(T) & =  \| \tilde{s}_{T}(x) - s(x)\|_{\LL}, \quad  \tilde{s}_{T}(x) \in \mathcal{V}^{(\mathcal{D})}_{\varphi}, 
\\
\varepsilon_{\mathcal{V}_{\varphi}^{(\mathcal{D})}}(T) & =  \| \tilde{s}_{T}(x) - s(x)\|_{\LL}, \quad  \tilde{s}_{T}(x) \in \mathcal{V}^{(\mathcal{D})}_{\varphi}, 
\\
\varepsilon_{\mathcal{V}}(T) & =  \| \tilde{s}_{T}(x) - s(x)\|_{\LL}, \quad  \tilde{s}_{T}(x) \in \mathcal{V}, 
\end{align}
where we have distinguished between the error terms associated to shift-invarinat and shift-variant basis of a generating function $\varphi$. 
\fi

\subsubsection*{Definition (Domain-Informed Generating Basis)}
For a given domain definition as in (\ref{eq:subdomains}), a basis $\Phi$ forms a domain-informed generating basis if and only if the following three conditions are met:
\begin{enumerate}
\item
$\Phi$ forms a Riesz basis, which is ensured if there exists constants $0 < A \le B< \infty$ such that \cite{ChristensenBook} 
\begin{align}
\label{eq:rieszbasis}
\forall  c \in \ell_{2}, \quad 0 & \le A \|c\|_{\ell_{2}}^{2} \le \bigg \| \sum_{k\in \mathbb{Z}}    c[k] \varphi_{k}(\frac{x}{T}-k) \bigg \|^{2}_{L_{2}} 
\nonumber
\\ 
& \le B \|c \|_{\ell_{2}}^{2} < \infty.
\end{align} 
This condition in necessary in order for (\ref{eq:representation}) to be a stable, unambiguous representation model.

\item
$\Phi$ forms a partition of unity, i.e., 
\begin{equation}
\label{eq:pouBasis}
\forall x \in \mathbb{R}, \quad \sum_{k\in \mathbb{Z}} \varphi_{k}(\frac{x}{T}-k) = 1. 
\end{equation}
This condition is necessary in order to have the approximation error vanish, see \cite[Appendix B]{Unser2000} for proof.

\item
Any element of $\Phi$ whose compact support lies entirely within a homogeneous segment of the domain becomes identical to its corresponding element in $\Phi^{\ast}$; i.e., $\varphi_{k}(x/T-k)=\varphi(x/T-k)$ if the domain is homogeneous at interval $[(k-\Delta^{(\varphi)})T, (k+\Delta^{(\varphi)})T]$.\\

\if{0}
for all $k \in \mathbb{Z}$, if there exits an $l \in \mathcal{J}$ such that
\begin{align}
d_{l}(x) &=1, \quad &\vert x-kT\vert \le \Delta^{(\varphi)}T,
\\
\{d_{j}(x)&=0\}_{j\in \mathcal{J}\setminus l}, \quad &\vert x-kT\vert \le \Delta^{(\varphi)}T, 
\end{align}
then $\varphi_{k}(x-k)=\varphi(x-k)$.\\
\fi

\end{enumerate}

\section{Domain-Informed B-spline Generating Basis} 
\label{sec:dibsib}

This section presents the construction of domain-informed generating bases using a B-spline generator function $\beta^{(n)}(x)$ of a desired order $n$. The construction framework is also applicable to alternative generating functions other than B-splines.

\subsection{B-spline Generating Basis}
The central B-spline $\beta^{(n)}(x)$ is obtained recursively as 
\begin{equation}
\label{eq:bsplines}
\beta^{(n)}(x) = (\beta^{(0)} \ast \beta^{(n-1)})(x),
\end{equation} 
where $(\cdot \ast \cdot)(x)$ denotes continuous-domain convolution and 
\begin{equation}
\beta^{(0)}(x) = 
\begin{cases}
1, \quad -\frac{1}{2}<x<\frac{1}{2}
\\
\frac{1}{2}, \quad |x| = \frac{1}{2}
\\
0, \quad \text{otherwise}.
\end{cases}
\end{equation}
\noindent 
$\beta^{(n)}(x)$ is supported in the interval $[ -\Delta^{(n)},\Delta^{(n)} ] $, where $\Delta^{(n)} = (n+1)/2$. We define $\Delta^{(n)}$-neighbourhood sets in the vicinity of each point $x$, denoted $\Delta^{(n)}_{x}$, as 
\begin{equation}
\label{eq:Delta}
\Delta^{(n)}_{x}: \Big \{k \in \mathbb{Z} \Bigm \vert \left |\frac{x}{T}-k \right | < \Delta^{(n)} \Big \};
\end{equation}
the set specifies the indices of sample points that fall within the $\Delta^{(n)}$-neighbourhood of a given point $x$. The scaled, integer-shifted set $\{ \beta^{(n)}(x/T-k)\}_{k \in \mathbb{Z}}$ forms a B-spline generating basis spanning the space of piecewise polynomial functions of order $n-1$. 

\subsection{Domain-Informed B-spline Generating Basis}
We construct a domain-informed B-spline associated to each sample point as the superposition of two B-spline based kernels as given in the following definition. \\

\subsubsection*{Definition (Domain-Informed B-splines)}
For a given set of subdomain functions (\ref{eq:subdomains}), a set of samples $s[k], k\in\mathbb{Z}$ and a B-spline generating function $\beta^{(n)}(x)$, a domain-informed B-spline associated to each sample point $k\in\mathbb{Z}$ can be defined as   
\begin{align}
\label{eq:displines}
\beta^{(n)}_{k}(x) 
& = 
\begin{cases}
\dot{\beta}^{(n)}_{k}(x) + {\ddot{\beta}}^{(n)}_{k}(x), & \quad |x| \le \Delta^{(n)}
\\
0, & \quad \text{otherwise},
\end{cases}
\end{align}
where $\dot{\beta}^{(n)}_{k}(x)$ denotes a dominant kernel that characterizes the overall shape and $\ddot{\beta}^{(n)}_{k}(x)$ a residual kernel that tunes the shape to the given domain knowledge in the adjacency of sample point $k$. \\

\noindent
In the following, we formulate the construction of $\dot{\beta}^{(n)}_{k}(x)$ and ${\ddot{\beta}}^{(n)}_{k}(x)$. Firstly, using the subdomain functions, a set of subdomain-informed B-splines associated to each $k \in \mathbb{Z}$ and $j\in\mathcal{J}$ are obtained as     
\begin{align}
\label{eq:sdiSPL}
\beta^{(n)}_{k,j}(x) 
&= 
\begin{cases}
d_{j}((x+k)T) \beta^{(n)}(x), \quad  & |x| \le \Delta^{(n)}
\\ 
0, \quad & \text{otherwise}. 
\end{cases}
\end{align}
It is straightforward to verify that subdomain-informed B-splines satisfy
\begin{equation}
\label{eq:sisProp1}
\sum_{k\in \mathbb{Z}} \beta^{(n)}_{k,j}(\frac{x}{T}-k)  = d_{j}(x), \quad \forall x \in \mathbb{R}, 
\end{equation}
and 
\begin{equation}
\label{eq:sisProp2}
\sum_{j\in\mathcal{J}} \beta^{(n)}_{k,j}(x)  = \beta^{(n)}(x), \quad \forall k \in \mathbb{Z}.
\end{equation}
Let $\mathcal{I}_{k}$ denote the set of indices of the subdomains maximally associated to sample $k$, i.e.,  
\begin{equation}
\mathcal{I}_{k}:\left \{i\in \mathcal{J} \Big \vert i = \argmax_{j\in \mathcal{J}} \left\{ d_{j}(kT) \right \}\right \},
\end{equation}
and $\mathcal{R}_{k}$ denote the set of the remaining subdomain indices, i.e., 
\begin{equation}
\mathcal{R}_{k} = \mathcal{J} \setminus \mathcal{I}_{k},
\end{equation}
where $\setminus$ denotes set difference. $\vert \mathcal{I}_{k} \vert > 1$ infers that more than one subdomain is maximally associated with sample point $k$. The dominant kernel $\dot{\beta}^{(n)}_{k}(x)$ is defined as
\begin{equation}
\dot{\beta}^{(n)}_{k}(x) = \sum_{i\in\mathcal{I}_{k}} \beta^{(n)}_{k,i}(x). 
\label{eq:dominantkernel}
\end{equation}
By cumulating the non-dominant subdomain-informed B-splines, a residual function $\Omega(x): x\in\mathbb{R} \rightarrow [0,1]$ can be defined as
\begin{align}
\label{eq:residualFunction}
\Omega(x) & = \sum_{l\in{\Delta^{(n)}_{x}}} \sum_{j\in \mathcal{R}_{l}}\beta^{(n)}_{l,j}(\frac{x}{T}-l)
\\ 
& = 1 -  \sum_{l\in{\Delta^{(n)}_{x}}} \dot{\beta}^{(n)}_{l}(\frac{x}{T}-l).
\label{eq:residualFunction2}
\end{align}
In particular, $\Omega(x)=0$ at any homogeneous part of the domain. Using $\Omega(x)$, the residual kernel $\ddot{\beta}^{(n)}_{k}(x)$ is defined as 
\begin{equation}
\label{eq:minorSPLs}
{\ddot{\beta}}^{(n)}_{k}(x) = 
\begin{cases}
\theta_{k}(x) \cdot \Omega((x+k)T), & \vert x \vert \le \Delta^{(n)}
\\
0, & \text{otherwise},
\end{cases}
\end{equation}
where $\theta_{k}(x)$, in its simplest form, equals a weighting function $w_{k}(x)$ defined as  
\begin{align}
\label{eq:wkx}
w_{k}(x) & = 
\begin{cases}
\displaystyle{\frac{\dot{\beta}^{(n)}_{k}(x)}{\sum_{l\in{\Delta^{(n)}_{x_{k}T}}} \dot{\beta}^{(n)}_{l}(x_{k}-l)}}, \quad & |x|\le \Delta^{(n)} 
\\
0, \quad & \text{otherwise},
\end{cases}
\end{align}
where $x_{k} = x+k$. The values of $\{d_{j}(x)\}_{j\in\mathcal{J}}$ are incorporated in $\dot{\beta}^{(n)}_{k}(x)$ and thus implicitly reflected in $w_{k}(x)$. To control changes in ${\ddot{\beta}}^{(n)}_{k}(x)$ relative to changes in the subdomains, $\theta_{k}(x)$ can be adaptively defined as 
\begin{equation}
\label{eq:theta}
\theta_{k}(x) := \displaystyle{\frac{\Theta\left(w_{k}(x)\right )}{\sum_{l\in{\Delta^{(n)}_{x_{k}T}}} \Theta\left (w_{l}(x_{k}-l)\right)}} \in [0,1],
\end{equation}
where $\Theta(\cdot): [0,1] \rightarrow [0,1]$ denotes a desired smooth, monotonically increasing function satisfying
$\Theta(x) = x$ for  $x\in\{0,0.5,1\}$, $\Theta(x)\le x$ for $x<0.5$ and  $\Theta(x)\ge x$ for $x>0.5$. Note that by setting $\Theta(x)=x$, we have $\theta_{k}(x) = w_{k}(x)$ since $\sum_{l\in{\Delta^{(n)}_{x_{k}T}}} w_{l}(x_{k}-l) = 1$. 

\begin{propo}(Domain-Informed B-spline Generating Basis) 
For a given sequence of samples $\{s[k]\}_{k \in \mathbb{Z}}$ as in (\ref{eq:samples}), a given set of subdomain functions as in (\ref{eq:subdomains}) and a given B-spline generating function $\beta^{(n)}(x)$, the set of functions $\{ \beta^{(n)}_{k}(x/T-k)\}_{k \in \mathbb{Z}}$ form a domain-informed generating basis, satisfying the three properties of a domain-informed generating basis.
\end{propo}
\begin{proof}
See Appendix I.
\end{proof}

\subsection{Illustrative Example Realization of Domain-informed B-spline Generating Bases}
\label{sec:example}
Two example realizations of domain-informed B-spline generating bases are presented to illustrate the construction scheme. Before presenting the constructions, a scheme is defined to randomly realize inhomogeneous domains that satisfy (\ref{eq:pouDomain}), consisting of a desired number of subbands with varying patterns of transition between subbands. A tunable $\Theta(x)$, cf. (\ref{eq:theta}), is also defined to enable varying the extent of adaptation of the bases to the domain knowledge.     

\subsection*{Realization of inhomogeneous domains}
An inhomogeneous domain consisting of $J$ subdomains can be realized through the following scheme. Let $L\in \mathbb{R}^{+}$ and $U \in \mathbb{R}^{+}$ denote the domain's lower and upper range, respectively. A system of $K \in \mathbb{Z}^{+}$ Meyer-type kernels $\{m_{k}(x), \: \forall x\in [L,U]\}_{k=1,\ldots,K}$ is constructed, see Fig.~\ref{fig:umt}(a). The details of the construction are given in Appendix~II. The benefit of using a Meyer-type system of kernels is that they provide smooth, compactly-supported kernels with overlap only between adjacent kernels. A random, monotonically increasing function $w(x): [L,U] \rightarrow [L,U]$ is then constructed. By incorporating $w(x)$ in $\{m_{k}(x)\}_{k=1,\ldots,K}$, a warped version of the Meyer-type system is obtained as: $\{m_{k}^{'}(x) = m_{k}(w(x))\}_{k=1,\ldots,K}$, see Fig.~\ref{fig:umt}(b). The warped kernels exhibit varying patterns of transition between adjacent kernels, with different orders of smoothness. 

\begin{figure}[htbp] 
   \centering
   \includegraphics[width=0.49\textwidth]{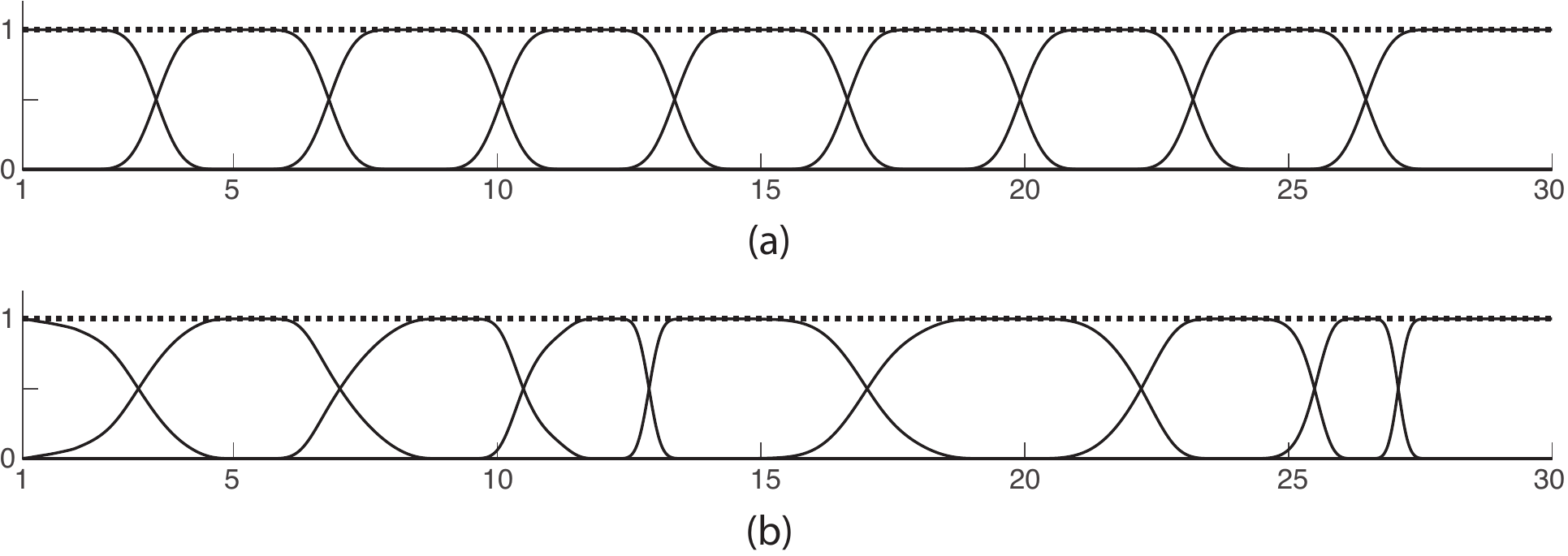} 
   \caption{(a) A Meyer-type system of kernels; $K=9, L=1$ and $U=30$. (b) A  warped version of the system of kernels in (a).}
   \label{fig:umt}
\end{figure}

A set of $J$ subdomain functions satisfying (\ref{eq:subdomains}) are then realized as
\begin{equation}
\forall x \in [L,U], \quad  d_{j}(x) = \sum_{k\in{\mathcal{K}_{j}}} m_{k}(x),\quad j=1,\ldots, J,  
\nonumber
\end{equation}
where $\mathcal{K}_{j} \subset \{1,\cdots,K\}$ such that $\cup_{j=1}^{J}\mathcal{K}_{j} = \{1,\cdots,K\}$ and $\cap_{j=1}^{J}\mathcal{K}_{j} = \emptyset$. Fig.~\ref{fig:diss_order1and3}(a) illustrates an inhomogeneous domain, consisting of two subdomains, constructed using this scheme. The domain has several homogeneous intervals, such as $[1, 6]$ and $[19,21]$, as well as varying inhomogeneous intervals present at the transition regions between the two subdomains. 

\begin{figure*}[htbp]
\begin{center}
\includegraphics[width=0.99\textwidth]{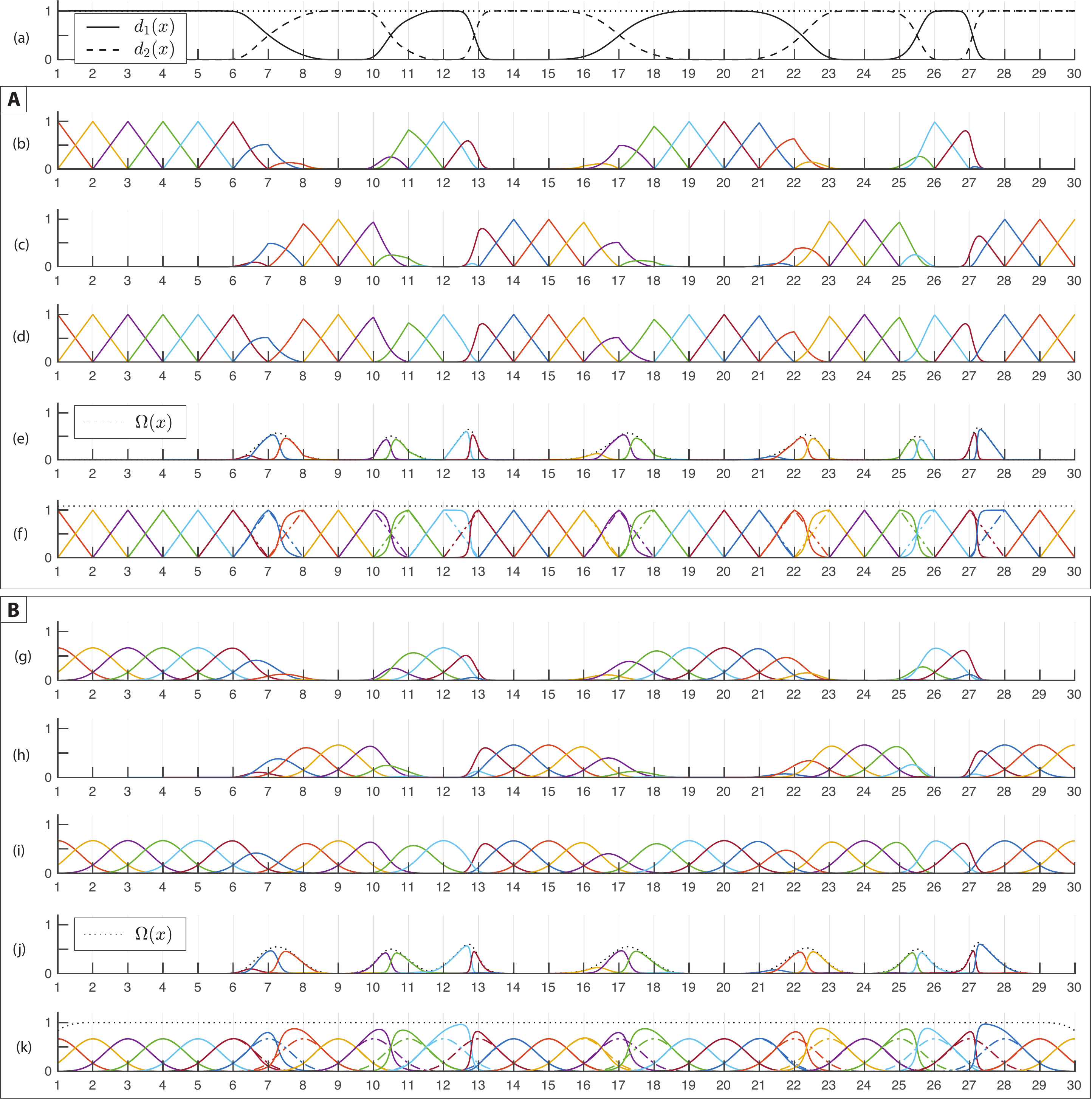}
\caption{Construction of domain-informed B-splines of order one (box A) and three (box B). (a) Realization of an inhomogeneous domain consisting of two subdomains, satisfying (\ref{eq:subdomains}). (b) Subdomain-informed B-splines $\beta^{(1)}_{k,1}(x-k)$. (c) Subdomain-informed B-splines $\beta^{(1)}_{k,2}(x-k)$. (d) Dominant kernels $\dot{\beta}_k^{(1)}(x-k)$. (e) Residual kernels $\ddot{\beta}_{k}^{(1)}(x-k)$. (f) Domain-informed B-splines of order one $\beta_{k}^{(1)}(x-k)$, constructed as the superposition of $\dot{\beta}_{k}^{(1)}(x-k)$ and $\ddot{\beta}_{k}^{(1)}(x-k)$. $\beta_{k}^{(1)}(x-k)$ (solid lines) are overlaid on their corresponding standard first order B-splines $\beta^{(1)}(x-k)$ (dashed lines). The dotted line shows the partition of unity property of the basis, cf. (\ref{eq:pouBasis}). (g)--(k) The same as (b)--(f), but for B-splines of order three. $\gamma=10$ in both the first order and third order designs.}
\label{fig:diss_order1and3}
\end{center}
\end{figure*}

\subsection*{Defining a tunable $\Theta(\cdot)$ function}
By incorporating a suitable $\Theta(\cdot)$ in (\ref{eq:theta}), the extent of adaptation to variation in proximity of subdomain transitions can be increased in the design of domain-informed B-splines. In particular, we exploit the logistic function to define $\Theta(\cdot)$ as 
\begin{equation}
\label{eq:ThetaLogistic}
\Theta(x) = (1+e ^{-\gamma/2})/(1+e ^{-\gamma (x-1/2)} ),
\end{equation}
where $\gamma \ge1$ is a free parameter. 

\subsection*{Realizations of B-spline Generating Bases}
In Fig.~\ref{fig:diss_order1and3}, sections A and B show constructions of domain-informed B-spline generating bases using first and third order B-spline kernels, respectively. Domain-informed B-splines whose support lie at a homogeneous part of the domain are identical to their standard counterpart; for example, see domain-informed B-splines centered at sample points 5, 15 and 20 in Figs.~\ref{fig:diss_order1and3}(e) and (i). On the other hand, domain-informed B-splines that reside in the adjacency of domain transition boundaries deviate from their standard counterpart. Domain-informed B-splines have their peak amplitude close to the sample point where they are localized and have declining tails as in standard B-splines. This property is due to the employed construction scheme that enables adaptation to domain inhomogeneities relative to the amplitude of the generating function along its support; i.e., the central parts of the kernels are more affected by domain inhomogeneities than their tails. 

\subsection{Domain-Basis Coherence}
It is insightful to quantify the coherence between a given inhomogeneous domain and a basis defined on the domain. To this aim, we first define a domain similarity metric.
      
\subsubsection*{Definition (Domain Similarity Metric)}
\label{def:2}
Given a description of an inhomogeneous domain as in (\ref{eq:subdomains}), a domain similarity metric can be defined in the $\Delta$-neighborhood of each $k \in \mathbb{Z}$ as 
\begin{equation}
\label{eq:dsm}
\xi_{k}(x) = 
\begin{cases}
\Theta  
\left ( 
\displaystyle{
1- \frac{1}{J} \sum_{j=1}^J \left| d_j(x_{k}T) - d_j(k T) \right|
} 
\right )
, &|x| < \Delta
\\
0, &|x| \ge \Delta,
\end{cases}
\end{equation}
where $x_{k} = x+k$, $\Theta(\cdot)$ is defined as in (\ref{eq:theta}); $\xi_{k}(x) \in [0,1]$. \\

The domain similarity metric can be used to quantify the relative difference in coherence between a domain and a domain-informed generating basis relative to that between the domain and the corresponding shift-invariant, generating basis.  

\subsubsection*{Definition (Domain-Basis Coherence Factor)}
Given a description of an inhomogeneous domain as in (\ref{eq:subdomains}), a 
 domain-basis coherence factor associated to domain-informed basis $\Phi$ can be defined as  
\begin{equation}
\mathcal{R}_{\Phi} = \frac{\sum_{k\in \mathbb{Z}} \langle \xi_{k}(x), \varphi_{k}(x)\rangle_{\text{L}_{2}}}{\sum_{k\in \mathbb{Z}} \langle \xi_{k}(x), \varphi(x)\rangle_{\text{L}_{2}}}.
\end{equation}

\begin{figure}[htbp] 
   \centering
   \includegraphics[width=0.48 \textwidth]{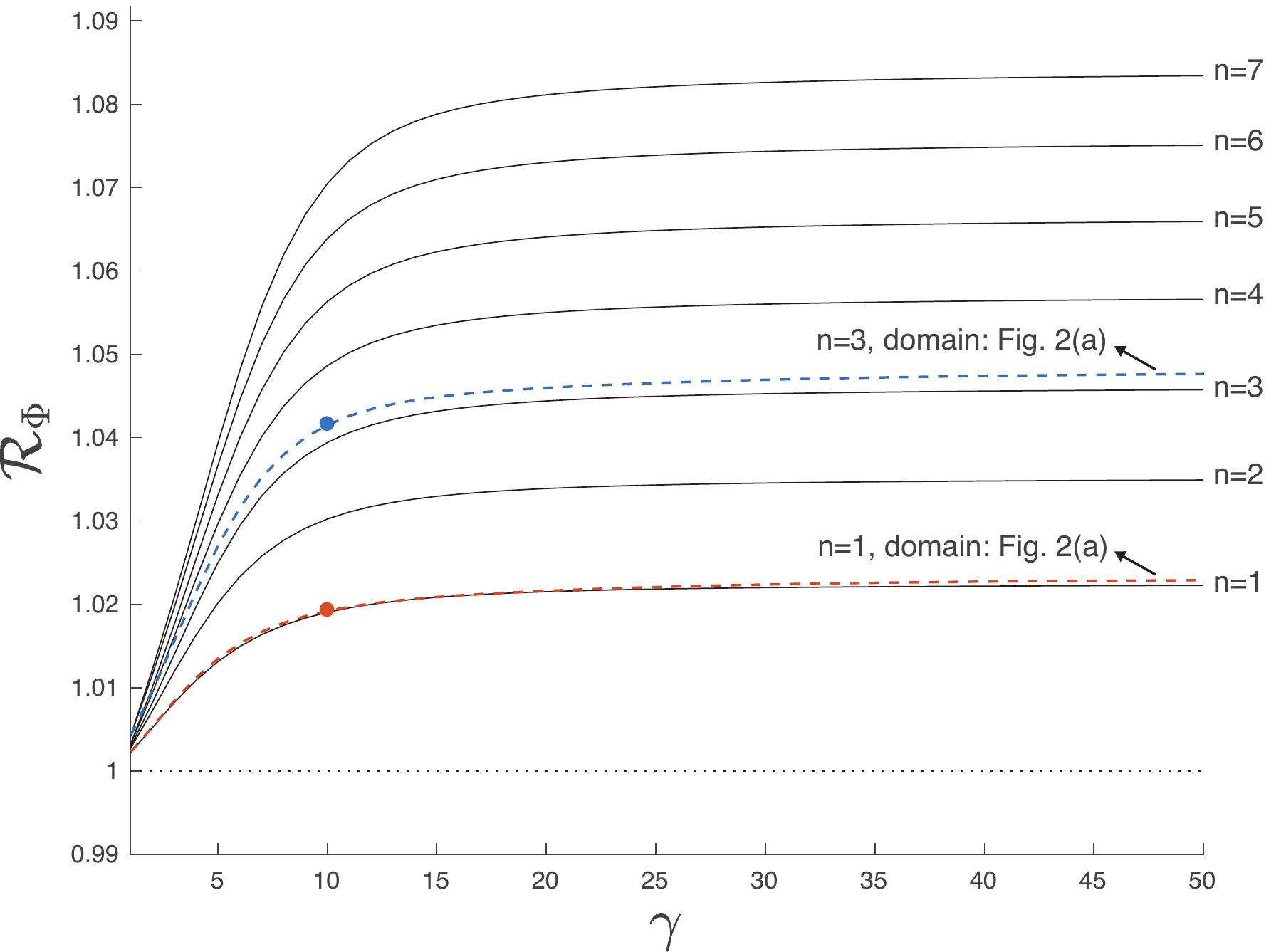} 
   \caption{Domain-basis coherence factor, for domain-informed B-splines of orders 1-7. The solid lines show ensemble values over 1000 inhomogeneous domain realizations. The dashed lines show the domain-basis coherence factor associated to the domain shown in Fig.~\ref{fig:diss_order1and3}; the two marked points show the domain-basis coherence factor for the particular domain-informed B-spline generating basis shown in Figs.~\ref{fig:diss_order1and3}(e) and (i), i.e., using $\gamma=10$.}
   \label{fig:db_coherence}
\end{figure}

In Fig.~\ref{fig:db_coherence}, the dashed curves show the change in $\mathcal{R}_{\Phi}$ for the domain shown in Fig.~\ref{fig:diss_order1and3}(a) as a function of the free parameter $\gamma$ in constructing domain-informed B-spline basis; the two marked positions on the dashed curves show $\mathcal{R}_{\Phi}$ associated to the domain-informed B-spline basis shown in Figs.~\ref{fig:diss_order1and3}(e) and (i). $\mathcal{R}_{\Phi}$ is greater than one across $\gamma$, ranging from 1 to 50 with steps of 1, reflecting a greater domain-basis coherence of the domain-informed B-spline basis relative to the standard B-spline basis. The coherence factor is greater for the third order domain-informed B-spline basis than that of the first order domain-informed B-spline basis. Moreover, $\mathcal{R}_{\Phi}$ increases with $\gamma$, up to a point where it almost saturates. Thus, $\gamma$ can be used to tune the level of adaptation of the design to the domain information relative to the extent of certainty associated with the domain data. The lower the adopted value of $\gamma$, the lower is the domain adaptation. 

To provide a more general picture of the coherence behavior, we generated 1000 realizations of inhomogeneous domains using the same scheme as that used to construct the domain shown in Fig.~\ref{fig:diss_order1and3}(a). Domain-informed B-spline basis of orders from 1 to 7 were constructed for each realization of inhomogeneous domain, over the range of $\gamma$. Fig.~\ref{fig:db_coherence} shows the resulting ensemble $\mathcal{R}_{\Phi}$, i.e., each point being the average of 1000 $\mathcal{R}_{\Phi}$ values, one associated to each domain realization. For each domain-informed B-spline basis order $n$, the ensemble coherence factor increases with the increase in $\gamma$, and gradually saturates. Moreover, for a given $\gamma$, the ensemble coherence factor associated to a higher order basis is greater than that associated to a lower order basis. {\color{black} In the following, all domain-informed B splines are generated with $\gamma=10$, which is the same as that used for the constructions shown in Fig.~\ref{fig:diss_order1and3}.}  

\section{Domain-Informed B-spline Interpolation}
We propose domain-informed B-spline interpolation (DIBSI) with the following assumption about the underlying signal: at inhomogeneous intervals of the domain, the signal is expected to be consistent with the given domain description, whereas at homogeneous intervals of the domain, the signal is expected to have smoothness characteristics provided by the chosen spline order $n$. \\ 

\begin{propo} 
\label{propo:propo1} (Domain-Informed B-spline Interpolation)
\noindent 
For a given domain description as in (\ref{eq:subdomains})--(\ref{eq:pouDomain}), a sequence of samples $\{s[k]\}_{k \in \mathbb{Z}}$ as in (\ref{eq:samples}) and a B-spline generating function $\beta^{(n)}(x)$, the $n$-th order domain-informed B-spline interpolant of the set of samples is obtained as
\begin{equation}
\label{eq:disi}
\forall x \in \mathbb{R}, \quad \tilde{s}(x) = \sum_{k\in \mathbb{Z}} c[k]  \beta^{(n)}_{k}(\frac{x}{T}-k),
\end{equation} 
where $\beta^{(n)}_{k}(\cdot)$ is given as in (\ref{eq:displines}) and the coefficients $c[k]$ are obtained through solving the set of equations 
\begin{equation}
\label{eq:solvecoeffs}
\forall k \in \mathbb{Z},\quad \sum_{l=k-\floor{n/2}}^{k+\floor{n/2}} c[l] \beta^{(n)}_{l}(k-l) = s[k].
\end{equation}
The resulting interpolant $\tilde{s}(x)$ has the following properties:
\begin{enumerate}
\item
$\tilde{s}(x)$ satisfies the consistency principle \cite{Unser1994}
\begin{equation}
\tilde{s}_{T}(kT) = s[k], \forall k \in \mathbb{Z}.
\end{equation}
 \item 
At homogeneous parts of the domain, $\tilde{s}(x)$ is equal to the $n$-th order spline interpolant obtained using standard spline interpolation, i.e., a piecewise polynomial of order $n-1$.  
 \end{enumerate}
\end{propo}

\begin{proof}
(Property 1) Domain-informed B-spline interpolation leads to perfect fit at sample points $k \in \mathbb{Z}$ since 

\begin{align}
\tilde{s}(k)  & = \left \langle \tilde{s}(x), \delta(x-kT) \right \rangle 
\nonumber
\\
& \stackrel{(\ref{eq:disi})}{=} \sum_{l \in \mathcal{K}} c[k] \underbrace{\left \langle \beta^{(n)}_{l}(\frac{x}{T}-l), \delta(x-kT) \right \rangle}_{\beta^{(n)}_{l}(k-l)} 
\nonumber
\\
& \stackrel{(\ref{eq:displines})}{=} \sum_{l=k-\floor{n/2}}^{k+\floor{n/2}}c[k] \beta^{(n)}_{l}(k-l)
\nonumber
\\
& \stackrel{(\ref{eq:solvecoeffs})}{=} s[k].
\end{align}

(Property 2) In any homogeneous domain interval $[a,b]$, domain-informed B-spline basis functions are identical to their standard B-spline basis function (Property 3 of domain-informed B-spline generating basis). Therefore, the domain-informed and standard B-spline interpolants become identical within the interval $[a,b]$.

\end{proof}

\subsubsection*{Remark} For domain-informed B-splines of order $n = 0$ and $1$, the coefficients $c[k]$ are identical to the signal samples, i.e., $\forall k \in \mathcal{K},\: c[k] = s[k]$. For higher order domain-informed B-splines, using matrix formulation, the solution to (\ref{eq:solvecoeffs}) is obtained by solving a $2\floor{n/2}+1$ band-diagonal matrix of linear equations; for cubic B-splines, i.e., $n=3$, the formulation leads to solving a tridiagonal matrix equation.

\begin{figure*}[tbp] 
   \centering
   \includegraphics[width=0.99\textwidth]{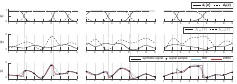} 
   \caption{\color{black} (a) Realizations of three inhomogeneous domains, each consisting of two subdomains; the construction is done using the scheme given in Section~\ref{sec:example}. (b) Realizations of pairs of random spline signals, cf. (\ref{eq:ftvj}). (c) Realizations of inhomogeneous signals obtained through integrating the random spline signals given in (b), cf. (\ref{eq:synthetic signals}), and interpolants obtained through interpolating the signals' samples using BSI and DIBSI of order three.}
   \label{fig:disi_synthetic_signals}
\end{figure*}

\section{DIBSI vs BSI on simulated data}
A set of synthetic random signals is required to evaluate the proposed interpolation scheme. Due to the assumption of inhomogeneity of the signal domain, realizing such signals is not straightforward. We present a scheme for realizing random signals of a desired piecewise smoothness characteristic on an inhomogeneous domain where the signals respect the inhomogeneity of the domain. Assume that an inhomogeneous domain is given as in (\ref{eq:subdomains}), consisting of $J$ subdomains. First, a set of $J$ signals of a desired smoothness characteristic are realized; in particular, each of the $J$ signals are constructed by realizing (i) a randomly jittered, uniform knot sequence {\color{black}$t[k]=k+\epsilon_{k}$}, where $k \in \mathbb{Z}$, $\epsilon_{k} \in[-\alpha,\alpha]$ and $\alpha \in [ 0,0.5)$, and (ii) $J$ random sequences of control values $\{v_{j}[k] \in [ 0,1] \}_{j=1}^{J}$. Using (i) and (ii), $J$ random spline signals of order $n$, denoted $f_{t,v_{j}}(x)$, are constructed such that \cite{DeBoorBook}
\begin{equation}
\forall k\in\mathbb{Z}, \quad  f_{t,v_{j}}(t[k]) = v_{j}[k], \quad j=1,\ldots, J.
\label{eq:ftvj}
\end{equation}
The signals $\{f_{t,v_{j}}(x)\}_{j=1}^{J}$ are then transformed to signals associated to each subdomain as 
$$
s_{j}(x) = H\left(d_{j}(x)-\frac{1}{J}\right) f_{t,v_{j}}(x), \quad j=1,\ldots, J,
$$
where $H(\cdot)$ denotes the Heaviside step function. With this model, we make a minimal assumption on the nature of the signals at the transition between subdomains. By superimposing the subdomain signals, a signal defined on the given inhomogeneous domain is obtained as 
\begin{equation}
 s(x) = \sum_{j\in \mathcal{J}} s_{j}(x).
 \label{eq:synthetic signals}
\end{equation}

Using this signal construction scheme and the domain construction scheme presented in Section~\ref{sec:example}, a set of $D$ domains, and on each domain a set of $S$ signals were randomly realized. B-spline interpolation (BSI) and DIBSI were then implemented on samples derived from each signal across a range of sampling steps $T \in [0.1,1]$ (step size = 0.1) and B-spline orders $n=1,\cdots,6$. {\color{black} Fig.~\ref{fig:disi_synthetic_signals} presents realizations of segments of three inhomogeneous domains, synthetic signals, and associated BSI and DIBSI interpolants, for $T=1$ and $n=3$. DIBSI generally better recovers the signal at inhomogeneous segments of the domain, through minimizing the contribution of signal samples that are not linked to the major subdomain at the interpolation point. For example, in the second domain realization, consider interval $(6,7)$, where the domain is dominantly associated to the first subdomain, i.e., $d_{1}(x)$. The BSI interpolant at this interval is heavily dependent on both adjacent samples, i.e., samples 6 and 7, whereas the DIBSI interpolant is mainly dependent on the value of sample 7, which is associated to $d_{1}(x)$. As another example, in the third domain realization, consider interval $(3,4)$, where the two subdomains transition almost at the center of the interval; BSI and DIBSI exhibit a similar behaviour in that they dominantly rely on sample 5 and sample 6 for obtaining the interpolant at the lower half and the upper half of the interval, respectively. In particular, BSI exhibits a balanced dependence on the sample value and the distance from the sample, whereas DIBSI exhibits a greater dependence on the sample value. Moreover, when subdomains transition at a point almost in between two samples, the DIBSI interpolant exhibits a pattern that reflects the underlying logistic function, cf.~(\ref{eq:ThetaLogistic}), used in constructing domain-informed B-splines. Fig.~\ref{fig:disi_gamma} illustrates how variation in the logistic function's $\gamma$ parameter is reflected in the DIBSI interpolant at transition bands. Finally, it is worth noting that if the signal exhibits a large variation in amplitude in the interval preceding the transition between subdomains, being inconsistent to the trend observed along the previous signal samples, neither BSI nor DIBSI can appropriately recover the signal; an example can be seen in interval $(2.5,3)$ in the third realization in Fig.~\ref{fig:disi_synthetic_signals}} 

\begin{figure}[b] 
   \centering
   \includegraphics[width=0.45\textwidth]{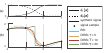} 
   \caption{\color{black} (a) Interval (4,7) of the third domain realization shown in Fig.~\ref{fig:disi_synthetic_signals} and (b) the associated synthetic signal, BSI and DIBSI interpolants of order three. DIBSI interpolants using three different $\gamma$ are illustrated.}
   \label{fig:disi_gamma}
\end{figure}

For $\tilde{s}_{T}(x)$ to be a good approximation of $s(x)$, the quality of interpolation needs to improve proportionally to the decrease in $T$. The interpolation error can be quantified using the distance metric $\| \tilde{s}_{T}(x) - s(x)\|_{\LL}$. The ensemble interpolation error associated with DIBSI can be quantified as 
\begin{equation}
\varepsilon_{\mathcal{V}_{\varphi}^{(\mathcal{D})}}(T) = \frac{1}{D S} \sum_{i=1}^{D} \sum_{j=1}^{S} \frac{\| \tilde{s}_{T}^{(i,j)}(x) - s^{(i,j)}(x)\|_{\LL}}{\|s^{(i,j)}(x)\|_{\LL}},
\label{eq:ensembleInterpError}
\end{equation}
where $s^{(i,j)}(x)$ denotes the $j$-th realized signal on the $i$-th realized domain and $\tilde{s}_{T}^{(i,j)}(x) \in \mathcal{V}_{\varphi}^{(\mathcal{D})}$ denotes the domain-informed B-spline interpolant obtained based on samples extracted from $s^{(i,j)}(x)$ at a sampling step $T$. For BSI where $\tilde{s}_{T}^{(i,j)}(x) \in \mathcal{V}_{\varphi}$, the ensemble interpolation error, denoted $\varepsilon_{\mathcal{V}_{\varphi}}(T)$, can be quantified in the same way as in (\ref{eq:ensembleInterpError}). 

\begin{figure}[tbp] 
   \centering
   \includegraphics[width=0.49\textwidth]{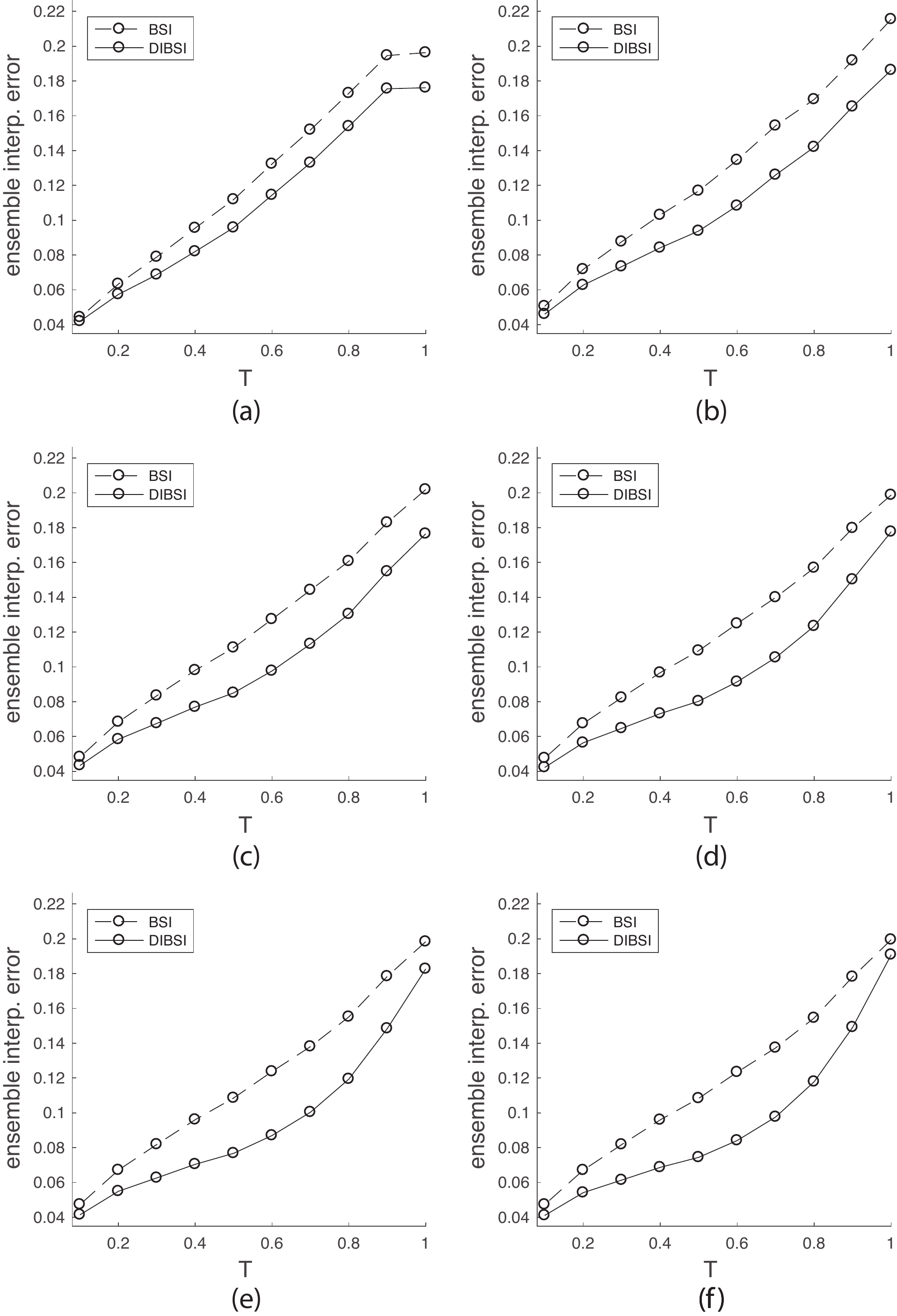} 
   \caption{(a)-(f) Ensemble interpolation errors of BSI and DIBSI, with generating basis constructed using $\beta^{(n)}(x)$ for $n =1,\ldots,6$, respectively. }
   \label{fig:disi_simulation_results}
\end{figure}

Fig.~\ref{fig:disi_simulation_results} shows ensemble interpolation errors $\varepsilon_{\mathcal{V}_{\varphi}}(T)$ and  $\varepsilon_{\mathcal{V}_{\varphi}^{\mathcal{D}}}(T)$ for $\varphi = \{\beta^{(n)}(x)\}_{n=1,\cdots,6}$, on a randomly realized set of 10,000 signals ($D=100$ and $S=100$). DIBSI outperforms BSI across the range of B-spline orders. For each B-spline order, both $\varepsilon_{\mathcal{V}_{\varphi}}(T)$ and  $\varepsilon_{\mathcal{V}_{\varphi}^{\mathcal{D}}}(T)$ decrease proportional to the decrease in $T$. The difference between $\varepsilon_{\mathcal{V}_{\varphi}}(T)$ and  $\varepsilon_{\mathcal{V}_{\varphi}^{\mathcal{D}}}(T)$ increases with the increase in the order of the B-splines. 

Aside from the comparisons showing the enhanced performance of DIBSI over BSI across the range of B-spline orders, it is insightful to compare the performance of DIBSI itself across different B-spline orders to that of BSI. Figs.~\ref{fig:disi_approximation_order}(a) and (b) show the approximation errors of BSI and DIBSI, respectively; the signal set used is the same as the one used for Fig.~\ref{fig:disi_simulation_results}, and thus, the plots are essentially a rearrangement of those shown in Fig.~\ref{fig:disi_simulation_results}. BSI using higher order B-splines does not consistently lead to greater reduction in the interpolation error, see Fig.~\ref{fig:disi_approximation_order}(a), whereas DIBSI shows better consistency, see Fig.~\ref{fig:disi_approximation_order}(b).

\begin{figure}[tbp] 
   \centering
   \includegraphics[width=0.49\textwidth]{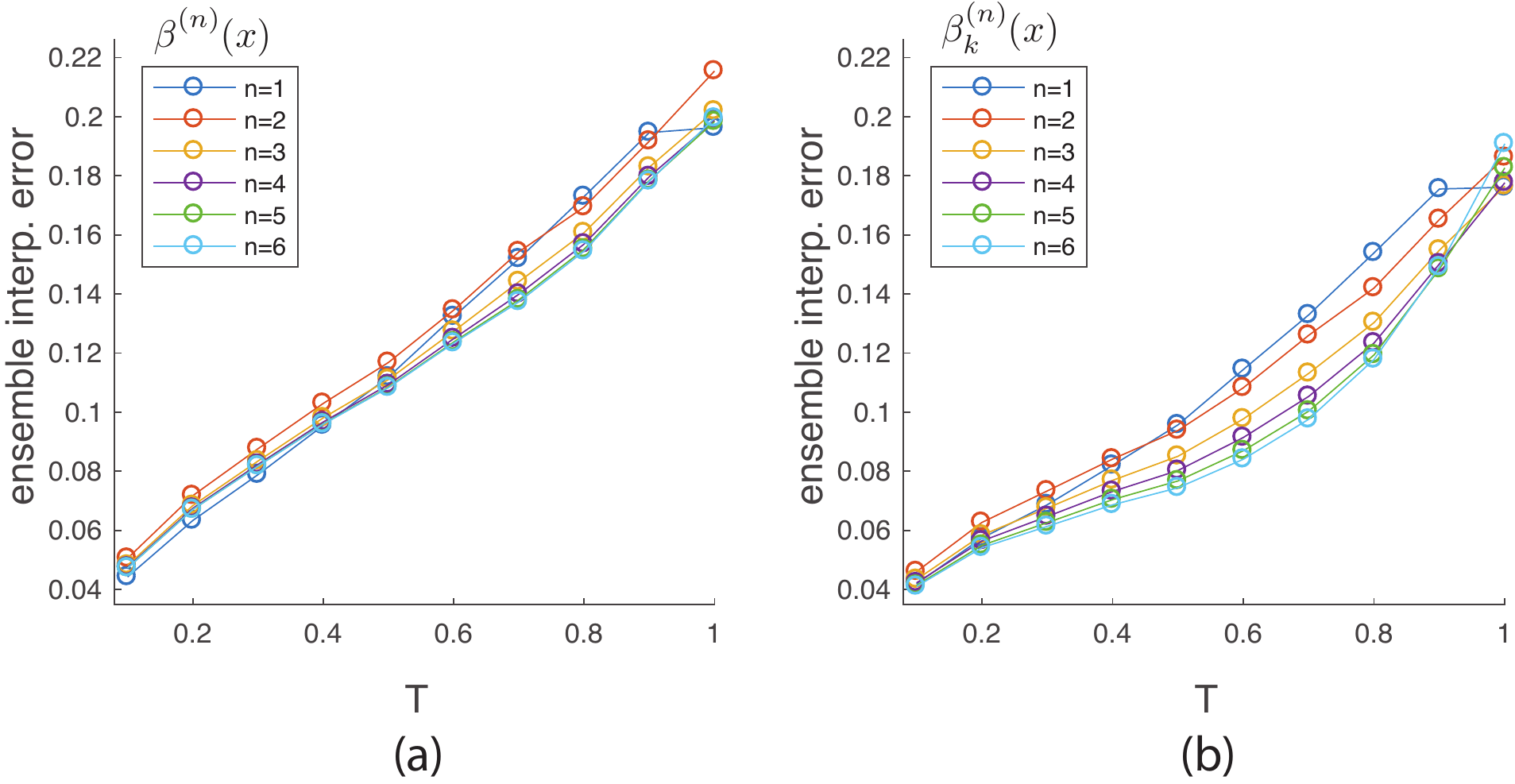} 
   \caption{Ensemble interpolation errors over signals realized on inhomogeneous domains $(D=100,S=100)$ using (a) BSI and (b) DIBSI.}
   \label{fig:disi_approximation_order}
\end{figure}

\begin{figure*}[htbp] 
   \centering
   \includegraphics[width=0.99\textwidth]{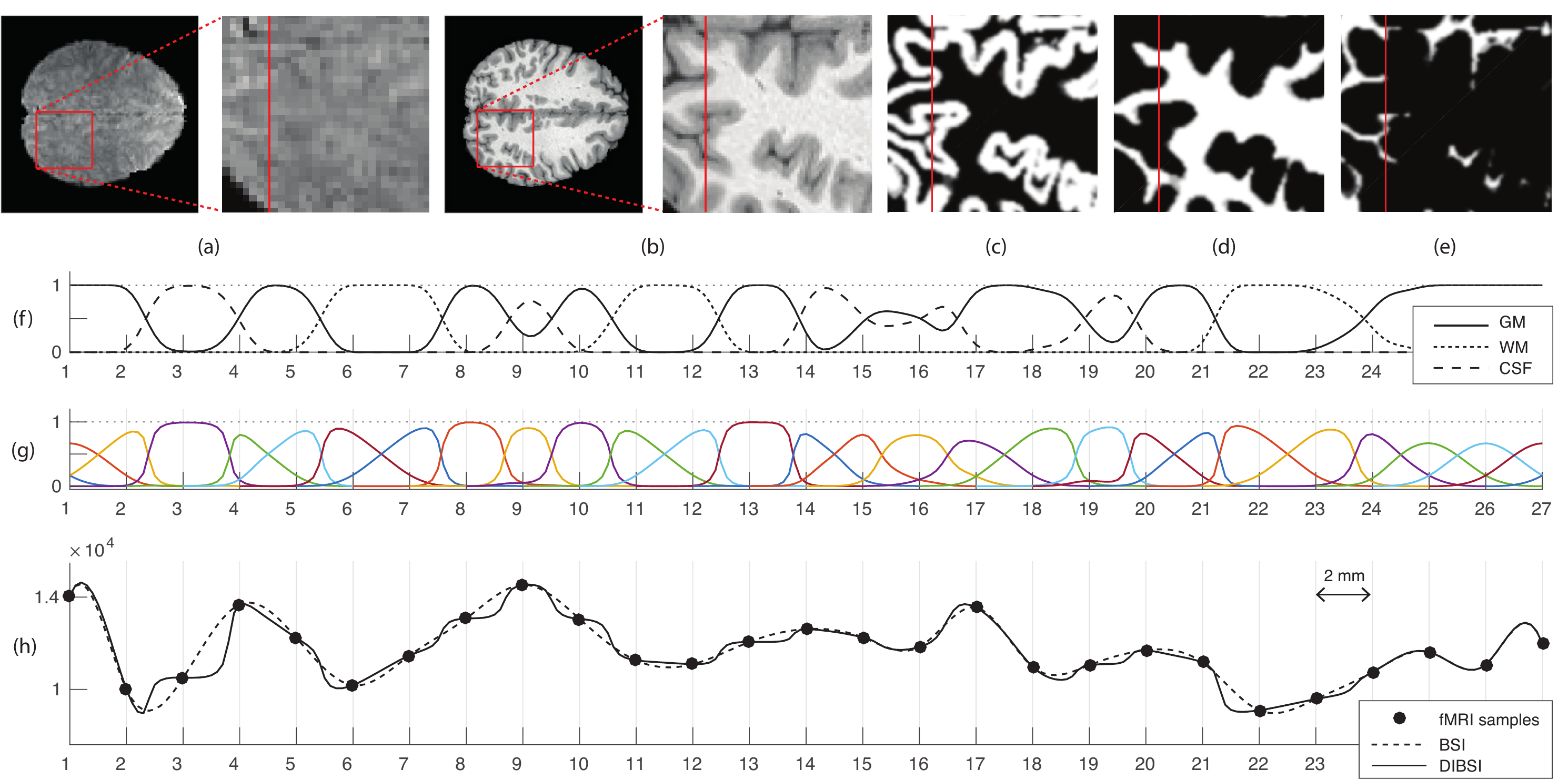} 
   \caption{(a) A slice of fMRI data of a subject, including a close-up of an ROI. (b) The subject's brain anatomy at the same neurological coordinate as in (a). (c) Gray matter, (d) white matter and (e) cerebrospinal fluid segmented probability maps of the ROI shown in (b). (f) Description of the inhomogeneous domain along the marked line within the ROI in (b), from top to bottom. (g) DIBSI basis, of order three, associated to the domain shown in (f). (h) B-spline interpolation and domain-informed B-spline interpolation of the functional samples along the marked line within the ROI in (a).}
   \label{fig:fmri}
\end{figure*}

\section{DIBSI vs BSI on neuroimaging data}
To show the practical significance of the proposed approach, we present interpolation results using the proposed scheme on an fMRI dataset. In brain studies using fMRI, a sequence of whole-brain functional data is acquired. To track brain activity at high temporal resolution, fMRI data are recorded at a relatively low spatial resolution. The data is commonly accompanied with a three- to fourfold higher resolution anatomical MRI scan, which provides information about the convoluted brain tissue delineating gray matter (GM) and white matter (WM), each of which have different functional properties \cite{Logothetis}, and cerebrospinal fluid (CSF). The topology of GM and WM varies across the brain as well as across subjects \cite{Mangin2004}. The goal is to exploit the richness of subject-specific anatomical data, which define the domain of the acquired fMRI data, for the purpose of improving the quality of interpolation and resampling of fMRI data. Such an interpolation scheme has an aim similar to signal processing techniques that enhance linear \cite{Kiebel2000} and non-linear \cite{Ozkaya2011, Behjat2015} denoising, deconvolution \cite{Farouj2017} and multi-scale decomposition \cite{Behjat2016} of fMRI data through exploiting anatomical constraints. Interpolation approaches that aim to map volumetric fMRI data on to the cortical surface \cite{Grova2006, Operto2008} are also related in the sense that they aim to enhance surface interpolation by accounting for the irregularity of the domain. Yet, such interpolation schemes are different from that we propose here since they adapt to the irregularity of the domain rather than to its inhomogeneity; i.e., they obtain an interpolant on a homogenous, irregular domain rather than one on an inhomogeneous domain.        

Fig.~\ref{fig:fmri} illustrates the setting for applying domain-informed B-spline interpolation on an fMRI dataset. We use data of a subject from the Human Connectome Project \cite{VanEssen2013}. Fig.~\ref{fig:fmri}(a) shows a 2-D slice of fMRI data, extracted from a 3D fMRI volume of the subject. The resolution of the image is $2\times2$ mm$^{2}$. Fig.~\ref{fig:fmri}(b) shows the structural scan of the brain of the same subject, which has an almost threefold higher resolution, $0.7\times0.7$ mm$^{2}$. Both slices are extracted such that they are aligned to the same neurological coordinate. By segmenting the anatomical scan, gray matter, white matter and cerebrospinal fluid probability maps, which determine the probability of each voxel being of either tissue, are obtained; Figs.~\ref{fig:fmri}(b)-(d) show these probability maps, upsampled to a resolution of $0.2\times0.2$ mm$^{2}$. We also treat the region outside the brain mask as part of the cerebrospinal fluid map. Bilinear interpolation of these maps are used to define normalized subdomain functions that satisfy (\ref{eq:subdomains}) along any column/row in the plane. Fig.~\ref{fig:fmri}(f) shows the subdomain function set along the marked line shown in Figs.~\ref{fig:fmri}(b)-(e). The domain is inhomogeneous, build of a convoluted mix of the two tissue types and CSF. 

Fig.~\ref{fig:fmri}(g) shows DIBSI basis of order three constructed for the inhomogeneous domain given in Fig.~\ref{fig:fmri}(f). DIBSI outperforms BSI across the range of B-spline orders. The basis is adapted to the convoluted description of the domain, and is robust to complex delineation patterns between subdomains, for instance see interval $[14, 17]$. Fig.~\ref{fig:fmri}(h) shows the fMRI samples along the marked line shown in Figs.~\ref{fig:fmri}(a), as well as the resulting BSI and DIBSI interpolants. At any homogeneous parts of the domain DIBSI results in an interpolant that is identical to BSI; for instance see the DIBSI and BSI interpolants within the interval [24, 27] in Fig.~\ref{fig:fmri}(g). At the inhomogeneous parts of the domain, DIBSI exhibits finer details than the {\color{black} BSI} image, for instance, see the domain description within the interval [7, 15], cf. Fig.~\ref{fig:fmri}(f). Within this interval, samples from gray matter (samples 8 and 10), white matter (samples 7 and 11) and cerebrospinal fluid (sample 9) are given, cf. Fig.~\ref{fig:fmri}(h). On the one hand, both BSI and DIBSI satisfy the consistency principle at the sample points. On the other hand, in between the samples, BSI maintains the smoothness characteristic enforced by using third order B-splines whereas DIBSI leads to a signal that is consistent with the description of the domain. For example, consider sample point 3 that is purely associated to CSF. Its adjacent samples, i.e., sample points 2 and 4, are  both associated to gray matter. In such cases, the domain-informed B-splines realized at the sample point extensively adapt to the associated domain, and as such, minimize the mixing of sample points associated to different subdomains. In other words, the best sample to use to obtain the interpolant within the range [2.5,4] is sample point 3, and therefore, the contribution of $\beta_{2}^{(3)}(x-2)$ and $\beta_{4}^{(3)}(x-4)$ is minimized within this range. Similar scenarios are observed at sample points 8 and 13. As another example, consider samples 6 and 7 that are both purely associated to white matter. Their adjacent samples, i.e., samples 5 and 8, lie within gray matter. In such a scenario, domain-informed B-splines exhibit a \textquoteleft crossing\textquoteright\, behavior. Although both $\beta^{(3)}_{6}(x-6)$ and  $\beta^{(3)}_{8}(x-8)$ have part of their support within the interval $[5.5,7.5]$, their amplitude within this interval is significantly suppressed to prevent mixing of gray matter samples with white matter samples in obtaining an interpolant within white matter. Similar scenarios can be observed along the domain. 

\begin{figure}[tb] 
   \centering
   \includegraphics[width=0.48\textwidth]{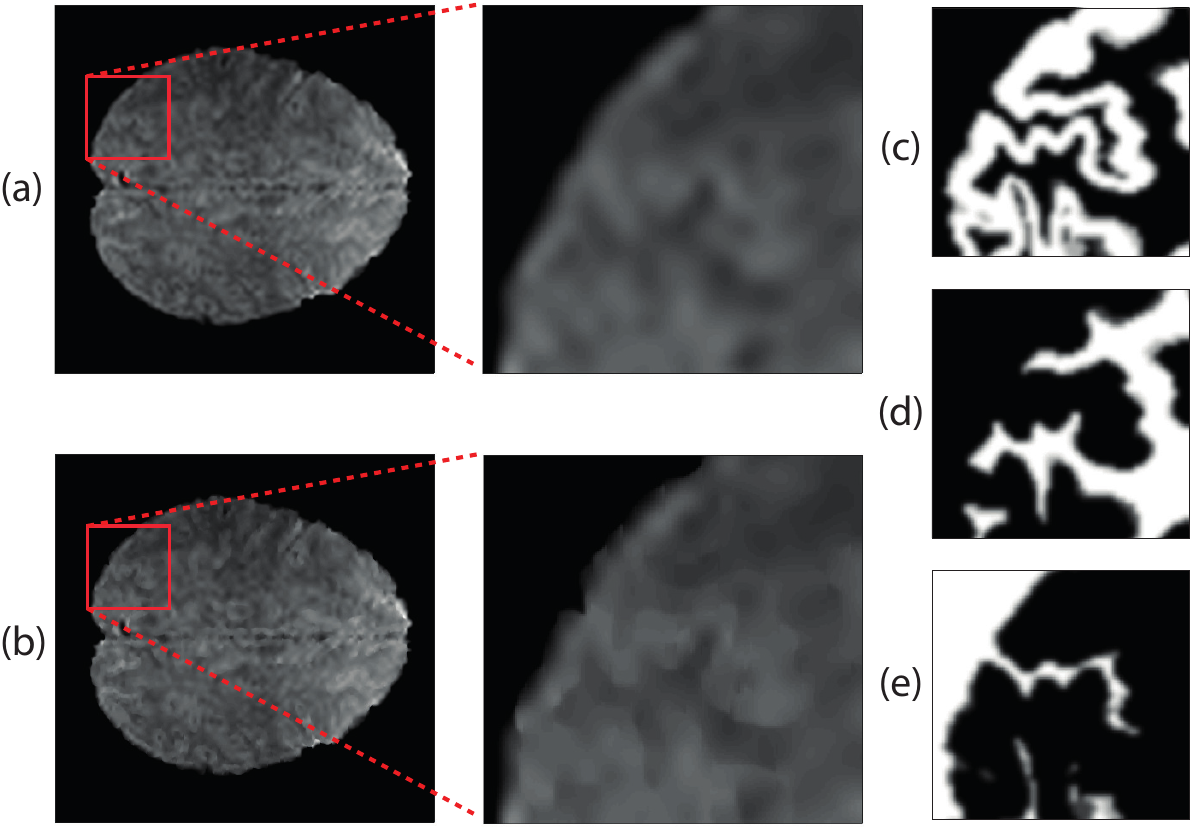} 
   \caption{(a) Upsampled slice of fMRI data using {\color{black} (a) BSI and (b) DIBSI}; a close up of an ROI is displayed. (c) Gray matter, (d) white matter and (e) cerebrospinal fluid segmented probability maps of the anatomy associated to the ROI shown in (a) and (b).}
   \label{fig:fmri_fullslice}
\end{figure}

A direct, separable extension of DIBSI to Euclidean m-D space can be formulated. We present results on interpolating the entire functional image shown in Fig.~\ref{fig:fmri}(a). In particular, we upsample the image by a factor of 10. For simplicity, the partial volume effect was disregarded and the value associated to each functional pixel was assigned to the center of the pixel. DIBSI was then performed along one dimension. The resulting upsampled values were then used for a second round of DIBSI along the second dimension. Figs.~\ref{fig:fmri_fullslice}(a) and (b) show upsampled versions of the functional slice using BSI and DIBSI, respectively; anatomical information associated to the ROI are shown in Figs.~\ref{fig:fmri_fullslice}(c)-(e). Overall, the upsampled image using DIBSI, cf. Fig.~\ref{fig:fmri_fullslice}(b), presents more details than that obtained using BSI, cf. Fig.~\ref{fig:fmri_fullslice}(a). For homogeneous parts of the domain, i.e., regions that fall purely within gray matter or white matter, BSI and DIBSI result in identical upsampled profiles. On the other hand, DIBSI results in finer detailed functional maps at inhomogeneous regions; for instance, see the maps along the sulcus (groove) shown in the ROI.     

It should be noted that the brain is intrinsically a 3-D structure, and that functional data are also acquired in 3-D. Therefore, extending the interpolation to 3-D and incorporating the association of the samples to the 3-D subdomain description can further enhance the results. Moreover, the exploited interpolation in Fig.~\ref{fig:fmri_fullslice} is a straightforward extension of DIBSI into 2-D to illustrate the use of such available domain information. A non-separable higher dimensional extension can provide more realistic encoding of convoluted domain descriptions such as that presented here for the fMRI setting. Such higher order extensions of DIBSI for neuroimaging applications will be explored in our future work. {\color{black} Finally, it is worth noting that the computational burden of DIBSI is extensive compared to BSI since it requires a domain-informed B-spline to be computed per sample point; in particular, the computational cost grows directly proportional to the number of samples that lie within or in adjacency of inhomogeneous intervals of the domain. Optimization of the implementation will be considered in future work.} 

\section{Conclusion}
We have proposed an interpolation scheme that incorporates a priori knowledge of the signal domain, such that the interpolant is consistent not only at sample points, with respect to the given samples, but also at intermediate points between samples, with respect to the given domain knowledge. The interpolation is formulated as an extension of B-spline interpolation. Shift invariant interpolation approaches that use generating functions other than B-splines may also be extended based on the scheme presented here. Results on simulated data showed reduced interpolation errors compared to using standard B-spline interpolation, on a range of B-spline orders and sampling steps. Results from applying the proposed approach on fMRI data demonstrated the potential of domain-informed interpolation in up-sampling low resolution functional brain data to obtain subtle activation patterns consistent to the anatomy of the brain.  

\section*{Appendix I}
In the following, we prove that the proposed domain-informed B-spline basis satisfies the three properties of a domain-informed generating basis.

\subsubsection*{Proof of property 1}
Let $\widehat{c}(e^{j\Omega})$ denote the discrete-domain Fourier transform of $c \in \ell_{2}$, i.e., $\widehat{c}(e^{j\Omega}) = \sum_{k\in \mathbb{Z}} c[k] e^{-j\Omega k}$ with the inverse transform given as
\begin{equation}
\label{eq:discreteFourierInverse}
c[k] = \frac{1}{2\pi}\int_{0}^{2\pi} \widehat{c}(e^{j\Omega})  e^{j\Omega k} d \Omega. 
\end{equation}
Let $\varphi_{k}(x) = \beta^{(n)}_{k}(x)$, 
and $\widehat{\varphi}_{k}(\omega) $ denote its continuous-domain Fourier transform given as 
\begin{align}
\widehat{\varphi}_{k}(\omega) 
& = \int_{-\infty}^{\infty} \varphi_{k}(x) e^{-j \omega x} \, d x.
\label{eq:contFourierTrans}
\end{align}
The objective is to derive bounds on the central term in (\ref{eq:rieszbasis}), which within this proof we denote by $s(x)$, i.e.,    
\begin{equation}
s(x) = \sum_{k\in \mathbb{Z}} c[k] \varphi_{k}(x-k).
\label{eq:centralterm}
\end{equation}
The continuous-domain Fourier transform of $s(x)$ is given as   
\begin{align}
\widehat{s}(\omega) 
& \stackrel{(\ref{eq:centralterm})}{=} \int_{-\infty}^{+\infty} \sum_{k\in\mathbb{Z}} c[k] \varphi_{k}(x-k) e^{-j\omega x} dx
\nonumber \\
& \stackrel{(\ref{eq:discreteFourierInverse})}{=}
\sum_{k\in \mathbb{Z}} \int_{0}^{2\pi} \widehat{c}(e^{j\Omega}) e^{ j\Omega k} \frac{d\Omega}{2\pi} \int_{-\infty}^{\infty} \varphi_{k}(x-k) e^{-j\omega x} dx
\nonumber \\
& = 
\int_{0}^{2\pi} \widehat{c}(e^{j\Omega}) \sum_{k\in \mathbb{Z}} \left(\int_{-\infty}^{\infty} \varphi_{k}(x-k) e^{-j\omega x} dx \right) e^{ j\Omega k} \, \frac{d\Omega}{2\pi},
\nonumber
\end{align}
which can be further reformulated by incorporating the change of variable $l = x-k$ and $d l = dx$ to 
\begin{align}
\widehat{s}(\omega) 
& =
\int_{0}^{2\pi} \widehat{c}(e^{j\Omega}) \sum_{k\in \mathbb{Z}} \left (\int_{-\infty}^{\infty} \varphi_{k}(l) e^{-j\omega l} dl \right) e^{-j\omega k} e^{j\Omega k} \, \frac{d\Omega}{2\pi} 
\nonumber \\
& \stackrel{(\ref{eq:contFourierTrans})}{=}
\frac{1}{2\pi}\int_{0}^{2\pi} \widehat{c}(e^{j\Omega}) \sum_{k\in \mathbb{Z}} \widehat{\varphi}_{k}(\omega) e^{-j\omega k} e^{j\Omega k} \, d\Omega
\nonumber \\
& = \frac{1}{2\pi}\int_{0}^{2\pi} \widehat{c}(e^{j\Omega}) {\color{black} \widehat{\Phi}(\omega,e^{j\Omega})} d\Omega, 
\label{eq:aux1}
\end{align}
where
\begin{equation}
{\color{black} \widehat{\Phi}(\omega,e^{j\Omega})} = \sum_{k\in \mathbb{Z}} \widehat{\varphi}_{k}(\omega) \, e^{-j \omega k} \, e^{j \Omega k}.
\label{eq:aux16}
\end{equation}

Invoking the $\LL$ space Parseval identity 
$
\| s \|_{\text{L}_{2}}^{2} = \frac{1}{2\pi} \int_{-\infty}^{\infty} \vert \widehat{s}(\omega) \vert^{2} d \omega
$
on (\ref{eq:aux1}) gives
\begin{align}
\|s \|^{2}_{L_{2}} 
& \stackrel{(\ref{eq:aux1})}{=} 
\frac{1}{(2\pi)^{3}}\int_{-\infty}^{\infty} \left \vert \int_{0}^{2\pi} \widehat{c}(e^{j\Omega}) \widehat{\Phi}(\omega,e^{j\Omega}) d\Omega \right \vert^{2} \diff\omega.   
\label{eq:aux3}
\end{align}
Using the Cauchy-Schwarz inequality, the integrand of the outer integral in (\ref{eq:aux3}) can be upper bounded as
\begin{align}
& \left \vert \int_{0}^{2\pi} \widehat{c}(e^{j\Omega}) \widehat{\Phi}(\omega,e^{j\Omega}) d\Omega \right \vert^{2} 
\nonumber
\\
& 
\le
\int_{0}^{2\pi} \vert \widehat{c}(e^{j\Omega}) \vert^{2} d\Omega \int_{0}^{2\pi} \vert \widehat{\Phi}^{\ast}(\omega,e^{j\Omega})  \vert^{2} d\Omega,
\nonumber
\\Ê
& 
= 
2 \pi \| c \|_{\ell_{2}}^{2} \int_{0}^{2\pi} \vert \widehat{\Phi}^{\ast}(\omega,e^{j\Omega})  \vert^{2} d\Omega,
\label{eq:cauchySchwarz}
\end{align}
where the equality follows from invoking the $\ell_{2}$ space Parseval identity 
$\| c \|_{\ell_{2}}^{2} = \frac{1}{2\pi} \int_{0}^{2\pi} \vert \widehat{c}(e^{j\Omega}) \vert^{2} d \Omega.$ 

{\color{black} Incorporating the outer integration and the constant factor in (\ref{eq:aux3}) on (\ref{eq:cauchySchwarz}) gives 
\begin{align}
& \frac{1}{(2\pi)^{3}}\int_{-\infty}^{\infty} \left \vert \int_{0}^{2\pi} \widehat{c}(e^{j\Omega}) \widehat{\Phi}(\omega,e^{j\Omega}) d\Omega \right \vert^{2} d\omega
\nonumber
\\
& 
\le 
\frac{1}{(2\pi)^{2}} \| c \|_{\ell_{2}}^{2}
\int_{-\infty}^{\infty}  
\int_{0}^{2\pi} 
\vert \widehat{\Phi}^{\ast}(\omega,e^{j\Omega})  \vert^{2} d\Omega \, d\omega
\nonumber
\\
&
=
\frac{1}{(2\pi)^{2}} \| c \|_{\ell_{2}}^{2}
\int_{0}^{2\pi} 
\underbrace{\int_{-\infty}^{\infty}  
\vert \widehat{\Phi}^{\ast}(\omega, e^{j\Omega})  \vert^{2} d\omega}_{=\| \widehat{\Phi}_{\Omega}(\omega) \|_{\text{L}_{2}}^{2}} \, d\Omega
\label{eq:integralOver2Pi}
\end{align}
where $\| \widehat{\Phi}_{\Omega}(\omega) \|_{\text{L}_{2}}^{2}$ denotes the $\text{L}_{2}$ norm of $\widehat{\Phi}(\omega,e^{j\Omega})$ over $\omega$ for a given fixed $\Omega$; note that $\forall \Omega \in [0,2\pi]$, $\| \widehat{\Phi}_{\Omega}(\omega) \|_{\text{L}_{2}}^{2}$ is bounded as 
\begin{align}
\| \widehat{\Phi}_{\Omega}(\omega) \|_{\text{L}_{2}}^{2}
& 
\stackrel{(\ref{eq:aux16})}{=}
\int_{-\infty}^{\infty}  
 \vert \sum_{k\in \mathbb{Z}} \widehat{\varphi}_{k}(\omega) \, e^{j \omega k} \, e^{-j \Omega k} \vert^{2} d\omega
 \nonumber
\\
& 
\le 
\int_{-\infty}^{\infty}  
 \sum_{k\in \mathbb{Z}} \vert \widehat{\varphi}_{k}(\omega) \, e^{j \omega k} \, e^{-j \Omega k} \vert^{2} d\omega
 \nonumber
\\
& 
= 
\sum_{k\in \mathbb{Z}} \int_{-\infty}^{\infty} \vert \widehat{\varphi}_{k}(\omega) \vert^{2} d\omega 
= 
2\pi \sum_{k\in \mathbb{Z}} \| \varphi_{k}(x) \|_{\text{L}_{2}}^{2},
\end{align}
where the summation runs over indices of the given finite set of signal samples; note that the norms $\| \varphi_{k}(x) \|_{\text{L}_{2}}$ exist since $\varphi_{k}(x) \in \text{L}_{2} $, which can be traced as follows
\begin{align}
\beta^{(n)}(x) \in \text{L}_{2} 
\nonumber
&
\stackrel{(\ref{eq:subdomains}),(\ref{eq:sdiSPL})}{\Longrightarrow} 
\beta^{(n)}_{k,j}(x) \in \text{L}_{2} 
\nonumber
\\
&
\stackrel{(\ref{eq:dominantkernel}),(\ref{eq:residualFunction})}{\Longrightarrow}
\dot{\beta}^{(n)}_{k}(x) \in \text{L}_{2}, \Omega(x) \in \text{L}_{2} 
\nonumber
\\
&
\stackrel{(\ref{eq:minorSPLs}),(\ref{eq:theta})}{\Longrightarrow} 
{\ddot{\beta}}^{(n)}_{k}(x)  \in \text{L}_{2} 
\nonumber
\\
&
~~\stackrel{(\ref{eq:displines})}{\Longrightarrow} 
\varphi_{k}(x) = \beta_{k}^{(n)}(x) \in \text{L}_{2}. 
\nonumber
\end{align}

Interpreting $\| \widehat{\Phi}_{\Omega}(\omega) \|_{\text{L}_{2}}^{2}$Êas function over $\Omega \in [0,2\pi]$, which is bounded and also nonnegative, its integral can be bounded by invoking the mean value theorem as 
\begin{align}
2\pi \inf_{\Omega\in[0,2\pi]}  \| \widehat{\Phi}_{\Omega}(\omega)  \|_{\text{L}_{2}}^{2} 
&
\le 
\int_{0}^{2\pi} 
\| \widehat{\Phi}_{\Omega}(\omega) \|_{\text{L}_{2}}^{2} d\Omega
\nonumber
\\
& \le 
2\pi \sup_{\Omega\in[0,2\pi]}   \| \widehat{\Phi}_{\Omega}(\omega)  \|_{\text{L}_{2}}^{2}.
\label{eq:boundsOnPhiNorm}
\end{align}
Incorporating (\ref{eq:aux3}) and the upper bound of (\ref{eq:boundsOnPhiNorm}) into (\ref{eq:integralOver2Pi}) gives
\begin{align}
\|s \|^{2}_{L_{2}} 
& 
\le
\underbrace{\frac{1}{2\pi} \sup_{\Omega\in[0,2\pi]}   \| \widehat{\Phi}_{\Omega}(\omega)  \|_{\text{L}_{2}}^{2}}_{B} \| c \|_{\ell_{2}}^{2},  
\end{align}
where $B$ denotes the required upper bound constant for declaring that $\{\varphi_{k}(x-k)\}_{k\in \mathbb{Z}}$ forms a Riesz basis, cf. (\ref{eq:rieszbasis}) and (\ref{eq:centralterm}). On the other hand, incorporating (\ref{eq:aux3}) and the lower bound of (\ref{eq:boundsOnPhiNorm}) into (\ref{eq:integralOver2Pi}) leads to one of the following to hold 
\begin{align}
& \overbrace{\frac{1}{2\pi} \inf_{\Omega\in[0,2\pi]}   \| \widehat{\Phi}_{\Omega}(\omega)  \|_{\text{L}_{2}}^{2}}^{A'} \| c \|_{\ell_{2}}^{2} \le \|s \|^{2}_{L_{2}} \le B \| c \|_{\ell_{2}}^{2}, 
\label{eq:scenario-a}
\\
& \|s \|^{2}_{L_{2}} < A' \| c \|_{\ell_{2}}^{2} \le B \| c \|_{\ell_{2}}^{2}. 
\label{eq:scenario-b}
\end{align}
 If (\ref{eq:scenario-a}) holds, then $A=A'$ is the required lower bound constant for declaring that $\{\varphi_{k}(x-k)\}_{k\in \mathbb{Z}}$ forms a Riesz basis. On the other hand, if (\ref{eq:scenario-b}) holds, since $A'$, $\| c \|_{\ell_{2}}^{2}$ and  $\|s \|^{2}_{L_{2}}$ are all strictly positive\footnote{We disregard the case where $||c||_{\ell_2}^{2} = 0$, since if $||c||_{\ell_2}^{2} = 0$, then $\forall k \in \mathbb{Z}, c[k] = 0$, then $\|s \|^{2}_{L_{2}}=0$, cf. (\ref{eq:centralterm}), and thus, the Riesz basis condition is satisfied $\forall A,B \in \mathbb{R}^{+}$, $A<B$.}, it can be stated that $\exists A \in (0,A')$ such that $A \| c \|_{\ell_{2}}^{2} \le \|s \|^{2}_{L_{2}}  < A' \| c \|_{\ell_{2}}^{2}\le B \| c \|_{\ell_{2}}^{2}$, thus, implicitly proving the existence of the lower bound constant $A$ and also providing an upper bound on A that is lower than B, i.e., $A<A^{'} \le B$.}

\subsubsection*{Proof of property 2}
We have, for all $x \in \mathbb{R}$ 
\begin{align}
& \sum_{k\in \mathbb{Z}} \beta^{(n)}_{k}(x-k)
\nonumber 
\\
& \stackrel{(\ref{eq:displines}),(\ref{eq:minorSPLs})}{=} 
\sum_{k\in \Delta^{(n)}_{x}} \dot{\beta}^{(n)}_{k}(x-k) + \sum_{k\in \Delta^{(n)}_{x}} \: \theta_{k}(x-k) \Omega(x)
\nonumber
\\
& \; \; \;  \stackrel{(\ref{eq:residualFunction})}{=}
\sum_{k\in \Delta^{(n)}_{x}} \sum_{i\in \mathcal{I}_{k}} \beta^{(n)}_{k,i}(x-k) 
\nonumber
\\ 
& \; \; \; \;  + \sum_{k\in \Delta^{(n)}_{x}} \theta_{k}(x-k) \sum_{l\in{\Delta^{(n)}_{x}}} \sum_{i\in \mathcal{R}_{k}}\beta^{(n)}_{l,i}(x-l)
\nonumber
\\
& \; \; \; \; =
\sum_{k\in \Delta^{(n)}_{x}}  \sum_{i\in \mathcal{I}_{k}} \beta^{(n)}_{k,i}(x-k) 
\nonumber
\\ 
& \; \; \; \; + \sum_{k\in{\Delta^{(n)}_{x}}} \sum_{i\in \mathcal{R}_{k}}\beta^{(n)}_{k,i}(x-k) \underbrace{\sum_{k\in \Delta^{(n)}_{x}} \theta_{k}(x-k)}_{=1}, 
\label{eq:pouTemp}
\\
&= \sum_{k\in \Delta^{(n)}_{x}}  \sum_{i\in \mathcal{J}} \beta^{(n)}_{k,i}(x-k) 
\stackrel{(\ref{eq:sisProp1}),(\ref{eq:sisProp2})}{=}1,
\label{eq:aux4}
\end{align}
where the partition of unity in (\ref{eq:pouTemp}) follows from  
\begin{align}
\sum_{k\in \Delta^{(n)}_{x}} \theta_{k}(x-k) 
& =
\displaystyle{\sum_{k\in \Delta^{(n)}_{x}}  \frac{\Theta\left(w_{k}(x-k)\right )}{\sum_{l\in{\Delta^{(n)}_{k+x-k}}} \Theta\left (w_{l}(k+x-k-l)\right)}}
\nonumber
\\
& 
= 1 .
\nonumber 
\end{align}

\subsubsection*{Proof of property 3}
For any $k\in\mathbb{Z}$, if the domain is homogeneous within the support $[ k-\Delta^{(n)}, k+\Delta^{(n)} ]$, i.e., there exists an $l \in \mathcal{J}$ such that for all $\vert x-k \vert \le \Delta^{(n)}$ we have $d_{l}(x) = 1$ and $\{d_{j}(x) = 0\}_{j\in \mathcal{J}\setminus l}$, then 
\begin{align}
\beta^{(n)}_{k}(x) 
& = \dot{\beta}^{(n)}_{k}(x) + \ddot{\beta}^{(n)}_{k}(x) 
\nonumber \\
& = \overbrace{d_{l}((x+k)T)}^{=1} \cdot \beta^{(n)}(x) 
+ \theta_{k}(x) \cdot \Omega(\overbrace{(x+k)}^{:=x_{k}}T), 
\nonumber \\
& \stackrel{(\ref{eq:residualFunction2})}{=}  \beta^{(n)}(x)
+ \theta_{k}(x) \left (1 -  \sum_{i\in{\Delta^{(n)}_{x_{k}T}}} \dot{\beta}^{(n)}_{i}(x_{k}-i)\right) 
\nonumber \\
& =  \beta^{(n)}(x)
+ \theta_{k}(x) \left (1 -  \overbrace{\sum_{i\in{\Delta^{(n)}_{x_{k}T}}} \beta^{(n)}(x_{k}-i)}^{=1}\right) 
\nonumber \\
& =  \beta^{(n)}(x)
\nonumber \\
\end{align}

\section*{Appendix II}
\label{sec:appendix_umt}
We define a Meyer-type system of $K\in \mathbb{Z}^{+}$ kernels $\{m^{'}_{k}(x): [L,U] \rightarrow [0,1]\}_{k=1}^{K}$, where $L, U\in \mathbb{R}^{+}$ and $L<U/(2K)$, in a way similar to that given in \cite{Behjat2016}, with the difference that i) we enforce the kernels to form a partition of unity as   
\begin{equation}
\label{eq:pouMeyer}
\sum_{k=1}^{K} \vert m^{'}_{k}(x) \vert   =1, \quad \forall x \in [L,U],
\end{equation}
as opposed to over their second power as in \cite{Behjat2016}, and ii) we skip enforcing the first and last kernels having $\text{L}_{2}$ norms equal to that of $\{m^{'}_{k}(x)\}_{k=2}^{K-1}$.

Using the auxiliary function of the Meyer wavelet \cite{Meyer} $\nu(x) = x^4(35-84x+70x^2-20x^3)$, a system $K \ge 2$ kernels $\{m^{'}_{k}(x)\}_{k=1}^{K}$ that span the range $[0,U]$ can be defined as 
\begin{subequations}
\label{eq:uniformFrame}
\begin{align}
m^{'}_{1}(x) & = 
\begin{cases}
\label{eq:kernelFirst}
1 & \quad \: \forall x  \in [L,\Delta/2]  \\
\cos^{2}(\frac{\pi}{2} \nu (\frac{x}{\Delta}-\frac{1}{2}))) & \quad \: \forall x \in (\Delta/2,3\Delta/2]  \\
0 & \quad \: \emph{elsewhere}
\end{cases}
\\ m^{'}_{k}(x) & = 
\begin{cases}
\label{eq:kernelMiddle}
\sin^{2}(\frac{\pi}{2} \nu (\frac{x}{\Delta}- k +\frac{3}{2}))  &  \forall x \in (\alpha,\alpha+\Delta]    \\
\cos^{2}(\frac{\pi}{2} \nu (\frac{x}{\Delta}- k+\frac{1}{2})) &  \forall x \in (\alpha+\Delta,\alpha+2\Delta]  \\
0 & \emph{elsewhere}
\end{cases}
\\ m^{'}_{K}(x) & = 
\begin{cases}
\label{eq:kernelLast}
\sin^{2}(\frac{\pi}{2} \nu (\frac{x}{\Delta}- K+\frac{3}{2}))  &  \forall x \in (\kappa,\kappa+\Delta]    \\
1 &  \forall x \in (\kappa+\Delta,U]  \\
0 & \emph{elsewhere}
\end{cases}
\end{align}
\end{subequations}
where $\Delta = U/K$ (see Fig.~\ref{fig:umt_notation} for notation), $\alpha  =(k-3/2) \Delta$ and $\kappa = (K-3/2) \Delta$. The set of kernels (\ref{eq:uniformFrame}) satisfy (\ref{eq:pouMeyer}) since 
\begin{align*}
\sum_{k=1}^{K} |m^{'}_{k}(\lambda)| 
& \: \: = \: \:
\begin{cases}
m^{'}_{1}(x) \: \stackrel{(\ref{eq:kernelFirst})}{=} 1 & \forall x \in [L,\Delta/2]
\\  m^{'}_{1}(x) + m^{'}_{2}(x) & \forall x \in (\Delta/2,3\Delta/2]
\\  m^{'}_{2}(x) + m^{'}_{3}(x) & \forall x \in (3\Delta/2,5\Delta/2]
\\ \vdots & \vdots
\\ m^{'}_{J}(x)\: \stackrel{(\ref{eq:kernelLast})}{=} 1 & \forall x \in (U-\Delta/2,U]
\end{cases} 
\\
& \stackrel{(\ref{eq:kernelMiddle})}{=} 
\begin{cases}
1 & \forall x \in [L,\Delta/2]
\\ \cos^{2}(x_{\text{I}}) + \sin^{2}(x_{\text{I}}) & \forall x \in (\Delta/2,3\Delta/2]
\\ \cos^{2}(x_{\text{II}}) + \sin^{2}(x_{\text{II}}) & \forall x \in (3\Delta/2,5\Delta/2]
\\ \vdots & \vdots
\\ 1 & \forall x \in (U-\Delta/2,U]
\end{cases} \nonumber \\
& \: \: = 1  \quad \forall x \in [0,U]
\end{align*} 
where $x_{\text{I}} = \frac{\pi}{2} \nu (\frac{x}{\Delta}-1/2) $ and $x_{\text{II}} = \frac{\pi}{2} \nu (\frac{x}{\Delta}-3/2)$; in the first equality we use the property that for all $k \in \{2,\ldots,K-1\}$ the supports of $m^{'}_{k-1}(x)$ and $m^{'}_{k+1}(x)$ are disjoint.

\begin{figure}[]
\centering
\includegraphics[width=.48\textwidth]{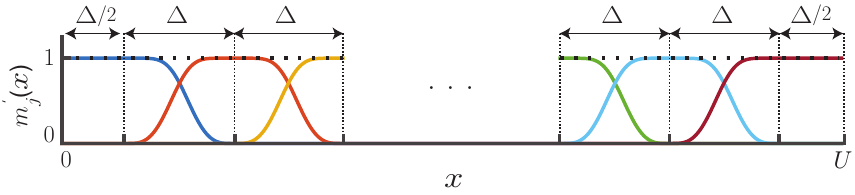} 
\caption[]{Meyer-type system of spectral kernels.}
\label{fig:umt_notation}
\end{figure}

\bibliographystyle{IEEEtran}
\bibliography{hbehjat_bibligraphy}
\end{document}